%% file: main.tex
\documentclass[11pt]{article}
\input{header}

\usepackage[final]{acl}

\usepackage{times}
\usepackage{latexsym}
\usepackage{amsmath}
\usepackage[T1]{fontenc}
\usepackage[utf8]{inputenc}
\usepackage{microtype}
\usepackage{inconsolata}

\usepackage{graphicx}

\title{From Citations to Contributions:\\LLM-Assisted Credit Scoring of Research Articles}

\author{
  Sana Ebrahimi \\
  University of Illinois Chicago \\
  \texttt{sana@uic.edu} \\
  \And
  Suraj Shetiya \\
  IIT Bombay \\
  \texttt{surajs@cse.iitb.ac.in} \\\And
  Abolfazl Asudeh \\
  University of Illinois Chicago \\
  \texttt{asudeh@uic.edu}\\
  }

\begin{document}
\maketitle

\begin{abstract}
Citation-based measures of scientific influence typically treat citations as uniform signals, ignoring the different roles that cited works play in a paper's contribution. 
We introduce {\em contribution-based credit scoring} for research articles: a structured citation analysis that decomposes a paper's credit between its own original contribution and the prior work it builds on. 
Motivated by a cooperative-game view of scientific credit, we propose the \emph{contribution tree}, a hierarchical framework that conserves importance across the document structure and separates original from citation-derived contribution. 
To make this framework scalable, we use LLMs as noisy comparative estimators of local importance. We further extend the model to article collections by propagating contributions through weighted citation graphs, yielding corpus-level contributions and normalized influence scores.
Our experiments suggest that our framework captures contribution signals beyond surface-level heuristics.
Our code is available at our \href{https://github.com/sanaebrahimi/Importance_Scoring/}{Github Repository\footnote{\href{https://github.com/sanaebrahimi/Importance_Scoring/}{https://github.com/sanaebrahimi/Importance\_Scoring/}}}. 
\end{abstract}

\input{intro}

\input{tree}
\input{llmasoracle}

\input{tech_vs_cite}

\input{citationgraph}

\input{exp}
\input{conclusion}

\vspace{-1mm}
\section*{Acknowledgments}
\vspace{-1mm}

This material is based in part upon work supported by the National Science Foundation under Award No. \href{https://www.nsf.gov/awardsearch/show-award/?AWD_ID=2348919}{2348919}, ANRF under
ANRF/ECRG/2024/004976/ENS,
and CloudBank computation (\href{https://www.xras.org/public/requests/263909-ACCESS-CIS251215}{NSF ACCESS allocation CIS251215}).

\newpage
\input{limitation}
\bibliography{ref}

\appendix
\section*{Appendix}
\input{appendix/Toc}
\input{appendix/related}
\input{appendix/warmup}
\input{appendix/axioms}

\input{appendix/llm-estimator}
\input{appendix/proof}
\input{appendix/topological}

\input{appendix/exp}

\input{appendix/extended-exp}
\input{appendix/human-noise}
\input{appendix/prompts}

\input{appendix/annotators}
\end{document}

%% file: header.tex
\usepackage{color}	
\usepackage{xspace}
\usepackage[table]{xcolor}
\usepackage[utf8]{inputenc}
\usepackage[table]{xcolor}
\usepackage{wrapfig}
\usepackage{tabularx}
\usepackage{longtable}
\usepackage{array}
\usepackage{ragged2e}  

\usepackage{booktabs}  
\usepackage{textcomp}

\usepackage{balance}
\usepackage{amssymb}
\usepackage{tikz}

\usetikzlibrary{arrows.meta}

\usepackage{cite}

\usepackage{graphicx}
\usepackage{enumerate}
\usepackage[hidelinks]{hyperref}
\usepackage[font={small,sf}]{caption}
\usepackage{floatflt}
\usepackage{lipsum}
\usepackage{multirow}
\usepackage{soul}
\usepackage{mdframed}
\usepackage{amsmath,amsfonts}
\usepackage{comment}
\usepackage{subcaption}
\usepackage{enumitem}
\usepackage{csquotes}
\usepackage{tabularx}
\usepackage{tabularx}
\usepackage{tikz}
\usetikzlibrary{decorations.pathreplacing}
\usetikzlibrary{matrix, positioning}
\usetikzlibrary{arrows.meta}
\usepackage{mathtools}
\usepackage[a-1b]{pdfx}  
\usepackage[T1]{fontenc}
\usepackage{microtype}
\usepackage{bbm}
\usepackage{fnpct}  
\usepackage{tcolorbox}
\usepackage[labelfont=bf]{caption}
\usepackage{subcaption}
\usepackage{makecell}
\usepackage{fixltx2e}
\usepackage{rotating,multirow}
\usepackage{etoolbox,xspace}
\usepackage{algorithmicx}
\usepackage{algorithm}
\usepackage{fontawesome}
\usepackage{algpseudocode}
\usepackage{wrapfig}
\usepackage{scalerel,amssymb}
\usepackage{booktabs}
\usepackage{multirow}
\usepackage{siunitx}
\usepackage{amsmath}
\usepackage{amsthm}

\graphicspath{{./figs/}}

\usepackage{siunitx}
\usepackage{algorithm}

\usepackage{tcolorbox}
\tcbuselibrary{breakable, skins, listings}

\tcbset{
  promptbox/.style={
    breakable,
    colback=gray!6,
    colframe=gray!40,
    fonttitle=\bfseries\small,
    coltitle=white,
    attach boxed title to top left={yshift=-2mm, xshift=4mm},
    boxed title style={colback=gray!55, rounded corners},
    left=4pt, right=4pt, top=6pt, bottom=4pt,
    fontupper=\small\ttfamily,
    before upper={\setlength{\parindent}{0pt}},
  },
  systemprompt/.style={promptbox, colframe=blue!35, colback=blue!4,
    boxed title style={colback=blue!50}},
  userprompt/.style={promptbox,  colframe=teal!40, colback=teal!4,
    boxed title style={colback=teal!55}},
}
\tcbset{
  foundationalbox/.style={
    guidelinebox,
    colframe=teal!40,
    boxed title style={colback=teal!55},
  },
}
\tcbset{
  guidelinebox/.style={
    breakable,
    enhanced,
    colback=white,
    colframe=blue!35,
    boxrule=0.6pt,
    fonttitle=\bfseries\small,
    coltitle=white,
    attach boxed title to top left={yshift=-2mm, xshift=4mm},
    boxed title style={
      colback=blue!55,
      rounded corners,
    },
    left=10pt,
    right=10pt,
    top=10pt,
    bottom=8pt,
    before skip=12pt,
    after skip=12pt,
    fontupper=\small,
    before upper={\setlength{\parindent}{0pt}},
  },
}
\usepackage{listings}
\newcommand{\commentstyle}{\frenchspacing \bfseries \# here is a comment: }

\DeclarePairedDelimiterX\set[1]\lbrace\rbrace{\,#1\,}

\usepackage{tikz}
\usetikzlibrary{positioning}

\newtheorem{theorem}{Theorem} 
\newtheorem{lemma}[theorem]{Lemma}

\usepackage{xspace}

\newcommand{\warrow}{\rightsquigarrow}

\newcommand{\abol}[1]{\textcolor{red}{abol: #1}\xspace}

\newcommand{\suraj}[1]{\textcolor{red}{Suraj: #1}\xspace}

%% file: intro.tex
\section{Introduction}\label{sec:intro}

Citations are the primary observable mechanism through which scientific articles acknowledge intellectual dependence on prior work. They support scholarly navigation, bibliometric indicators, citation-network analysis, and downstream measures of scientific influence. Yet most citation-based measures treat {\em citations as uniform signals}: a cited paper receives the same unit of credit whether it provides broad background, a methodological building block, or the central idea for the citing paper. This abstraction enables large-scale analysis, but it ignores the functional role of the citations.

Prior work\footnote{Related work is further discussed in Appendix~\ref{sec:related}.} has sought to enrich this view by incorporating the syntactic and semantic context of citations, including through citation function classification, sentiment analysis, citation summarization, and citation-based retrieval and recommendation~\cite{cca}. These approaches, however, primarily characterize citation contexts rather than quantify contributions. Recent LLM-based approaches assign holistic research-quality scores from limited surface text such as titles and abstracts~\cite{paper-eval}, but remain black-box document-level predictors rather than structured mechanisms for attributing contribution across sections, paragraphs, and references.

Therefore, in this paper, we study {\em contribution-based credit scoring}: decomposing a research article's credit between its own original contribution and the prior work it builds on. This problem has a natural {\em cooperative-game} interpretation: a paper and its references jointly form the final scientific work, and a principled allocation assigns credit according to marginal contribution, with the Shapley value providing a canonical solution concept.

Despite its conceptual appeal, this cooperative-game formulation is not directly operational. Shapley-style allocation requires counterfactual papers formed from arbitrary subsets of references, but such documents are unobserved.

We therefore replace counterfactual coalition evaluation with an observation-based structural attribution framework. The key observation is that the document itself provides structure: scientific articles have an inherent hierarchy that constrains how contributions can be distributed. We formalize this structure as a {\em contribution tree}, which conserves importance across the document hierarchy and supports a decomposition of the paper's credit into original and citation-derived contribution.
This allocation satisfies natural analogues of the Shapley principles, including efficiency, additivity, and hierarchical conservation. We then compute original and citation-derived contributions bottom-up by identifying citation blocks and aggregating their scores to the referenced works.

Making this framework operational requires a scalable way to assess the relative contribution of document components. 
While human experts could provide such judgments, eliciting them across papers, document units, and citation contexts would be prohibitively expensive. This motivates our use of LLMs as contribution estimators. Recent studies have explored LLMs as tools for supporting scientific assessment and peer review~\cite{deepreview,metareviewer}; more broadly, progress in LLM-assisted reviewing has led major computer science conferences, including NeurIPS 2026\footnote{\href{https://neurips.cc/Conferences/2026/ai-reviewing-experiment}{NeurIPS 2026 AI-Assisted Reviewing Experiment}.}, to explore opt-in AI-assisted reviewing workflows. We do not treat the LLM as a substitute for expert judgment, but as a {\em noisy} comparative estimator of local importance.

Finally, we extend the framework from individual papers to article collections. We construct a weighted citation graph in which each edge reflects the contribution assigned to a cited work by the citing paper. Credit is then propagated through this graph, so that a paper can receive contribution not only from papers that cite it directly, but also through indirect citation paths. We provide an efficient algorithm based on topological ordering that computes corpus-level contribution scores without enumerating citation paths. We further define a normalized influence score that removes each paper's self-contribution and measures its share of the cross-paper influence mass in the corpus.

We evaluate our framework on research papers annotated by human experts and API-based LLMs (Claude Sonnet-4.6 and GPT-OSS-120B). The results show that surface-level heuristics, such as citation frequency are insufficient for citation-level contribution scoring, while small local models such as {\tt qwen3:1.7b} perform reasonably well. Ablation studies and a corpus-level case study further support the effectiveness of our approach in separating influence from original contribution.

%% file: tree.tex
\input{figs/tree}

\section{Contribution Tree}\label{sec:tree}

\vspace{-2mm}
The contribution-based credit scoring of a research article $d$  can be formalized as a cooperative game with $d$ and its citations $\mathcal{C}(d)$ as the players, where $d$ captures the paper's own {\em original
contribution} and each $c\in\mathcal{C}(d)$ captures a {\em citation contribution} (Appendix~\ref{sec:warmup}).
Although this formulation, combined with the Shapley-value allocations~\cite{shapley1953value}, is conceptually natural and provides a principled theoretical foundation for our problem, it is impractical as it requires unobservable counterfactual
papers over subsets of references.

We therefore shift to {\em a structural framework} that operates directly on the paper's observable content.
We represent a scientific document as a rooted tree $T = (V, E)$, which we call the \emph{contribution tree} (Fig.~\ref{fig:cont-tree}). The root node $d \in V$ corresponds to the whole paper, and its importance is normalized as
\(
  \sigma(d) = 1.
\)
Each internal node $u \in V$ corresponds to a structural unit of the paper (a section, subsection, or subsubsection), and each directed edge $(u, v) \in E$ connects a parent unit to a child unit in the document hierarchy. The terminal structural units are paragraphs, which form the leaves of the structural skeleton.

We assign each node $v \in V$ a non-negative importance score $\sigma(v) \geq 0$. To each edge $(u, v) \in E$ we associate a weight $w(u, v) \in [0,1]$ representing the fraction of $u$'s importance that flows to child $v$. Importance is conserved at every
level. That is,
$\sum_{v \in \mathrm{ch}(u)} w(u,v)=1$, for every internal node $u$. Hence,

\vspace{-11mm}
\begin{align}\label{eq:computeScore}
   \hspace{8mm}\nonumber \sigma(v) &= w(\mathrm{par}(v),\, v)\cdot \sigma(\mathrm{par}(v)) 
   \\
   &= \prod_{u\in \text{path }r\warrow v} w(\mathrm{par}(u),\, u),
\end{align}

\vspace{-3mm}
where $\mathrm{par}(v)$ is the parent of $v$ and $\mathrm{ch}(u)$ represents the children of $u$.
In Fig.~\ref{fig:cont-tree}, for example, consider {\tt\small [Paragraph 1]}, under the {\tt\small [Introduction]} node.
The score of this paragraph is computed as
\(\sigma(\text{\tt\small Intro.Paragraph 1})=w_1\times w_{11}.\)
By induction, the scores of the root's children sum to $1$, and at every deeper level the scores of siblings sum to their parent's score. 
As a result, in Fig.~\ref{fig:cont-tree}, \(\sum_{i=1}^5w_i=1\), and \(w_{211}+w_{212}=1\).

Next, we study the key properties of the contribution tree. 
Although not exact counterparts, these properties are analogous to Shapley-value axioms.
\begin{itemize}[leftmargin=*,
  itemsep=0pt,
  topsep=0pt,
  parsep=0pt,
  partopsep=0pt]
    \item {\em Hierarchical conservation}: For any internal node $u$, the contribution of the parent is the sum of the importance scores of the children.
    \item {\em Efficiency}: The total contribution of an article is fully distributed across the leaves.
    \item {\em Additivity}: 
    For any finite collection of disjoint node sets $\{S_i\}_{i=1}^m$, $\sigma\!\left(\bigcup_{i=1}^m S_i\right)=\sum\sigma(S_i)$.
\end{itemize}

\begin{theorem}\label{th:axiom}
    Contribution Tree satisfies {efficiency}, {additivity}, and {hierarchical conservations}.
\vspace{-3mm}
\end{theorem}

\begin{proof}
    See Appendix~\ref{app:axioms}.
    \vspace{-2mm}
\end{proof}

%% file: figs/tree.tex
\newcommand{\wgt}[1]{\textcolor{red!70!black}{$#1$}}

\begin{figure}[t]
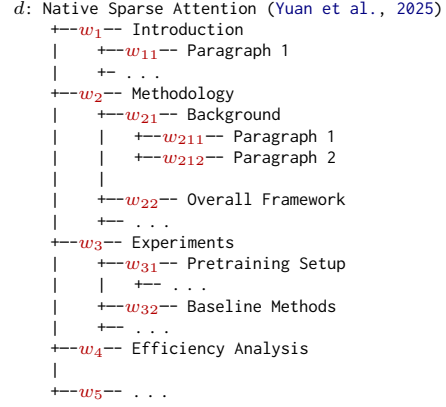

\centering
\scriptsize
\ttfamily
\begin{tabular}{@{}l@{}}
$d$: Native Sparse Attention~\cite{yuan2025native}\\
~~~~+----\wgt{w_1}---- Introduction\\
~~~~|~~~~+----\wgt{w_{11}}---- Paragraph 1\\
~~~~|~~~~+-- \ldots\\
~~~~+----\wgt{w_2}---- Methodology\\
~~~~|~~~~+----\wgt{w_{21}}---- Background\\
~~~~|~~~~|~~~+----\wgt{w_{211}}---- Paragraph 1\\
~~~~|~~~~|~~~+----\wgt{w_{212}}---- Paragraph 2\\
~~~~|~~~~|\\
~~~~|~~~~+----\wgt{w_{22}}---- Overall Framework\\
~~~~|~~~~+---- \ldots\\
~~~~+----\wgt{w_3}---- Experiments\\
~~~~|~~~~+----\wgt{w_{31}}---- Pretraining Setup\\
~~~~|~~~~|~~~+---- \ldots\\
~~~~|~~~~+----\wgt{w_{32}}---- Baseline Methods\\
~~~~|~~~~+---- \ldots\\
~~~~+----\wgt{w_4}----  Efficiency Analysis\\
~~~~|\\
~~~~+----\wgt{w_5}---- \ldots
\end{tabular}
\vspace{-3mm}
\caption{Illustration of a Contribution Tree.}
\label{fig:cont-tree}
\vspace{-6mm}
\end{figure}

%% file: llmasoracle.tex
\subsection{Measuring the Edge Weights}
\label{sec:weights}
The edge weights $w(u,v)$ of the contribution tree are not observed directly. They are latent quantities induced by the content of the research article: the ideas expressed, the evidence provided, and the functional role played by each segment within its parent unit. The central challenge is therefore to estimate these weights from the text alone.


To formalize the weight measurement, we rely on a {\em contribution-estimation oracle}. Given a parent node $u$ and its children $\mathrm{ch}(u)$, the oracle returns a non-negative sibling-relative score $X(v \mid u)$ for each child $v \in \mathrm{ch}(u)$. These scores are intended to reflect the relative importance of the children within the local context of $u$. Since the oracle scores need not be calibrated across different parent nodes, we recover edge weights by normalizing only within each sibling set:

\vspace{-8mm}
\begin{equation}
\label{eq:local-normalization}
  \hspace{10mm}\widehat{w}(u,v)
  =
  \frac{X(v \mid u)}
  {\sum_{v' \in \mathrm{ch}(u)} X(v' \mid u)}
\end{equation}

\vspace{-2mm}
\noindent Thus, the oracle is required only to produce comparative non-negative scores. 

In practice, any instantiation of the oracle may be noisy. As a result, repeated queries for the same parent node may produce different sibling scores. We model the raw oracle response as

\vspace{-4mm}
\begin{equation}
\label{eq:estimation_with_noise}
  X(v \mid u)
  =
  \alpha(u) w(u,v) + \eta_{u,v},
\end{equation}

\vspace{-2mm}
\noindent where $\alpha(u) > 0$ is an arbitrary local scale factor and $\eta_{u,v}$ captures response noise. The scale factor reflects the fact that raw oracle scores are meaningful only within a sibling set: multiplying all scores under the same parent by a positive constant does not change the normalized allocation in Equation~\ref{eq:local-normalization}. This motivates estimating each local allocation from multiple noisy oracle responses.

Let $X^{(t)}(v \mid u)$ denote the raw score returned for child $v$ in the $t$-th query to the oracle for parent node $u$. For each valid response\footnote{A response is valid if its weight allocations are non-negative.}, we first normalize across siblings:

\vspace{-9mm}
\begin{equation}
\label{eq:sample-normalization}
  \widehat{w}^{(t)}(u,v)
  =
  \frac{X^{(t)}(v \mid u)}
  {\sum_{v' \in \mathrm{ch}(u)} X^{(t)}(v' \mid u)} 
\end{equation}

\vspace{-2mm}
Let $\mathcal{A}(u)$ denote the set of valid responses for node $u$. The final edge-weight estimate is the average normalized allocation

\vspace{-6mm}
\begin{equation}
\label{eq:aggregate-estimator}
  \widehat{w}(u,v)
  =
  \frac{1}{|\mathcal{A}(u)|}
  \sum_{t \in \mathcal{A}(u)}
  \widehat{w}^{(t)}(u,v)
\end{equation}

\vspace{-3mm}
This aggregation reduces variance while preserving local conservation: for every internal node $u$, the estimated outgoing weights remain non-negative and sum to one.

\vspace{-2mm}
\paragraph{Assumption.}
We require the contribution-estimation oracle to perform better than a random guess estimator with no information, which uniformly distributes the weights between the children, i.e.,
\(
  w_0(u,v) = \frac{1}{|\mathrm{ch}(u)|}.
\)
Ideally, a contribution estimator should satisfy

\vspace{-7mm}
\[
  \mathbb{E}\Big[
  |\widehat{w}(u,\cdot) - w(u,\cdot)|
  \Big]
  <
  |w_0(u,\cdot) - w(u,\cdot)|
\]

\vspace{-2mm}\noindent
meaning that its expected error is smaller than the error of assigning equal importance to all children. Since the true latent weights $w(u,v)$ are unobserved, we evaluate this condition empirically using human annotations as a proxy for ground truth in Section~\ref{sec:eval}.

In principle, the contribution estimates could be elicited from human experts who have read and understood the paper. However, expert annotation is slow and costly, making it difficult to scale. This motivates an automated instantiation of the contribution-estimation oracle.
In particular, as further explained in Appendix~\ref{app:llm}, we adopt the {\bf LLM-as-Estimator} model, in which LLMs take the role of the contribution-estimation oracles.

%% file: tech_vs_cite.tex
\vspace{-2mm}
\subsection{Original vs. Citation Contribution}
\label{app:technical-citation-contribution}

\vspace{-1mm}
The contribution tree provides a normalized allocation of importance across the structural hierarchy of a research article. In particular, the score $\sigma(v)$ of a node $v$ quantifies the total contribution assigned to the corresponding structural unit, such as a subsection. However, as discussed above, our main objective is to separate the paper's own original contribution from the contribution derived from its cited references.
Specifically, following Equation~\ref{eq:contribution}, for every node $v$ in the contribution tree, we seek to decompose its total score $\sigma(v)$ into two components: its own original contribution $\sigma_{\mathrm{orig}}(v)$, and the citation contribution $\sigma_{\mathrm{cite}}(v)$ of the references cited under $v$. Thus, for each node $v$, we require 
\(\sigma(v) = \sigma_{\mathrm{orig}}(v) + \sigma_{\mathrm{cite}}(v).\)

While the computation of the total node score $\sigma(v)$ is performed top-down through the contribution tree, as reflected in Equation~\ref{eq:computeScore}, the decomposition into original and citation-derived contribution is computed {\em bottom-up}. 

Consider a leaf node $v$ in the contribution tree, corresponding to a paragraph. If the paragraph does not include any citations, then the entire contribution of the paragraph is attributed to the paper itself. In this case, the original contribution of the node is
\(
\sigma_{\mathrm{orig}}(v) = \sigma(v),
\)
while
\(
\sigma_{\mathrm{cite}}(v) = 0.
\)


When the paragraph includes one or more citations, we first identify the {\em ``citation blocks''} within the paragraph. A citation block is a window of text in the paragraph that discusses the contribution, role, or relevance of one or more cited works. 
Intuitively, a citation block captures the portion of a paragraph whose contribution is attributed to cited work rather than to the paper's own original content.
After identifying citation blocks,\footnote{We use LLMs to identify citation blocks as text spans surrounding one or more citations.} we attach them as children of the corresponding paragraph nodes with an additional node corresponding to the text that does not belong to a citation block.

Let $\mathcal{B}(p)=\{b_1,\ldots,b_k\}$ denote the citation blocks contained in a paragraph $p$. The original contribution of $p$ is the part of its total score that is not assigned to its citation blocks, i.e., 
\(
\sigma_{\mathrm{orig}}(p)
=
\sigma(p)-\sum_{i=1}^{k}\sigma(b_i).
\)
Equivalently, 

\vspace{-4mm}
\begin{equation*}
\sigma_{\mathrm{cite}}(p)
=
\sum_{i=1}^{k}\sigma(b_i).
\end{equation*}

\vspace{-2mm}
For an internal node $v$, the decomposition is computed bottom-up by aggregating the corresponding scores of its children. That is,

\vspace{-7mm}
\begin{align*}
\sigma_{\mathrm{orig}}(v)
&=
\sum_{u\in \mathrm{ch}(v)} \sigma_{\mathrm{orig}}(u),
\\
\sigma_{\mathrm{cite}}(v)
&=
\sum_{u\in \mathrm{ch}(v)} \sigma_{\mathrm{cite}}(u).
\end{align*}

\vspace{-2mm}
Thus, the total {\bf original contribution} of the article is given by $\sigma_{\mathrm{orig}}(d)$, where $d$ is the root of the contribution tree.

For each paragraph $p$ containing citation blocks, the LLM estimator receives $\sigma_{cite}(p)$ and determines how this contribution should be distributed across these blocks within $p$. Let $\sigma(b)$ denote the contribution assigned to citation block $b$. It remains to allocate the score of each citation block to the references appearing in it. Let $\mathcal{C}(b)$ denote the set of references cited in citation block $b$. We distribute the block score uniformly among these references, assigning each $c\in \mathcal{C}(b)$ the local contribution
\(
\sigma_b(c)=\frac{\sigma(b)}{|\mathcal{C}(b)|}.
\)
Finally, the {\bf citation contribution} of each reference $c\in \mathcal{C}(p)$ is obtained by summing its contributions over all citation blocks within paragraph $p$ in which it appears:

\vspace{-4mm}
\begin{align}\label{eq:citation-contribution}
\sigma_p(c)
=\sum_{b \in p:~c\in \mathcal{C}(b)} \sigma_b(c) 
\end{align}


%% file: citationgraph.tex
\section{From Individual-level to Article Collections}\label{sec:citationgraph}

\vspace{-2mm}
We next extend the article-level contribution scoring framework to a corpus of research articles. Let $\mathcal{P}$ be a collection of papers within a specific research topic, such as {\tt [Sparse Attention]}. For each article $d \in \mathcal{P}$, the contribution tree decomposes the normalized credit of $d$ between its own original content and the references it cites. This gives a local view of the scoring, centered on a single paper $d$.

\vspace{-2mm}
\subsection{Corpus-level Contribution Score}
To move from the local perspective to the corpus-level, we consider the {\em citation graph} (aka. citation network) of $\mathcal{P}$. A citation graph, also called a citation network, is a directed graph whose nodes represent articles and whose edges represent citation relations between them~\citep{egghe1990introduction,zhao2015analysis}. Formally, we define
\(
G_{\mathcal{P}} = (\mathcal{P}, E),
\)
where \((d\rightarrow q) \in E\) if article $d$ cites the prior work $q$. 

Due to temporal nature of citations, citation graphs are commonly modeled as {\em directed acyclic graphs} (DAG)~\citep{clough2015transitive}. We adopt this convention and treat \(G_{\mathcal{P}}\) as a DAG with no cycle\footnote{For exceptional cases caused by factors such as versioning, the cycles are broken arbitrarily.}.

The temporal structure of the citation graph induces {\em a natural direction for credit propagation}: credit flows from a paper to the earlier papers that it cites, and recursively to papers reachable through longer citation paths.
This enables extending the contribution tree from an individual-document attribution model to a domain-level credit-flow model. 
The contribution tree determines how much of a paper's credit is assigned to each immediate reference. The citation graph then propagates this assigned credit along paths in \(G_{\mathcal{P}}\).
Thus, a paper may receive credit not only from articles that cite it directly, but also indirectly through chains of citation-mediated dependence.

\input{figs/dag}

We extend the citation graph by assigning a weight to every edge. For each edge
\((d \rightarrow q) \in E\), we set its weight to \(\sigma_d(q)\), where
\(\sigma_d(q)\) is the citation contribution assigned to reference \(q\) by the
contribution tree of article \(d\). Thus, the edge weight is not a raw citation
count, but a contribution-based score measuring how much of \(d\)'s credit is
attributed to \(q\).

For example, Fig.~\ref{fig:toy-citation-graph} shows contribution propagation
in a toy citation graph with papers \(\{a,b,c,d,e,f,g\}\). The edge
\(a \rightarrow f\) indicates that article \(a\) cites article \(f\), and its
weight \(\sigma_a(f)\) is the citation contribution of \(f\) to \(a\).

The weighted citation graph induces a recursive credit redistribution process.
Each article initially receives one unit of credit. A fraction
\(\sigma_{\mathrm{orig}}(d)\) of the credit received by article \(d\) is retained
as the article's own original contribution, while the remaining citation-derived
credit is redistributed to its references according to the edge weights
\(\sigma_d(c)\), for \(c \in C(d)\). Consequently, credit propagates along
directed paths in \(G_{\mathcal{P}}\). 

Let \(\{(a\warrow d)\}_{\mathcal{P}}\) denote the set of directed paths in \(G_{\mathcal{P}}\). The amount of credit from \(a\)
that reaches \(d\) along a path \(\pi\) is the product of the edge weights on
that path, out of which $\sigma_{\mathrm{orig}}(d)$ is retained by $d$ and the rest is propagated to its references.
Therefore, the corpus-level contribution of an article
$d\in\mathcal{P}$ is 

\vspace{-7mm}
\begin{align}
\label{eq:corpus_contribution}
\sigma_{\mathcal{P}}(d)
=
&\sigma_{\mathrm{orig}}(d)
\Big( 1 + 
\\
\nonumber & \sum_{\text{path}\in \{(a\warrow d)\}_{\mathcal{P}}}
\prod_{(u\rightarrow v)\in \text{path}}
\sigma_u(v)
\Big)
\end{align}

\vspace{-3mm}
\begin{theorem}\label{th:1}
Let \(\mathcal C_{\mathcal P}(d)=\mathcal C(d)\cap \mathcal P\) denote the citations of \(d\) that belong to the corpus. Suppose that, after restricting the citation graph to \(\mathcal P\), each paper satisfies \(
    \sigma_{\mathrm{orig}}(d)
    +
    \sum_{c\in \mathcal C_{\mathcal P}(d)} \sigma_d(c)
    =
    1
\). Then, \(
    \sum_{d\in\mathcal P}\sigma_{\mathcal P}(d)=|\mathcal P|.
\)
\vspace{-2mm}
\end{theorem}

\begin{proof}
    See Appendix~\ref{app:proof}.
    \vspace{-2mm}
\end{proof}

In Appendix~\ref{app:topological}, we provide an efficient algorithm for computing the corpus-level contribution scores based on a topological ordering~\cite{kleinberg2006algorithm} of the citation graph.

\vspace{-2mm}
\subsection{Normalized Influence Score}

The corpus-level contribution score $\sigma_{\mathcal{P}}(d)$ measures the total credit retained by article $d$ within the corpus $\mathcal{P}$. This total includes the unit of credit assigned to $d$ by
itself, through its own original contribution $\sigma_{\mathrm{orig}}(d)$.

For some analyses, however, we are interested specifically in the extent to which an article {\em influences other papers} in the corpus. In this case, the self-contribution term can obscure the signal of interest. This issue is particularly pronounced when the corpus is small or when the citation graph is sparse or weakly connected: only a limited amount of contribution mass may reach a paper through citation paths, so the self-retained original credit can dominate the total corpus-level score.

To isolate the credit that $d$ receives from other articles, we define the external contribution mass of $d$ as

\vspace{-7mm}
\begin{align*} \sigma_{\mathcal{P}\backslash d}(d) &=\sigma_\mathrm{orig}(d)\sum_{\pi\in \{a\warrow d\}_{\mathcal{P}}} \prod_{(u\rightarrow v)\in \pi} \sigma_u(v) \\ &= \sigma_{\mathcal{P}}(d) - \sigma_{\mathrm{orig}}(d) \end{align*}

\vspace{-2mm}
We then define the normalized influence score of $p$ as

\vspace{-12mm}
\begin{align}\label{eq:norm_influence}
I_{\mathcal{P}}(d)
=
\frac{
\sigma_{\mathcal{P}\setminus d}(d)
}{
\sum_{d'\in\mathcal{P}}
\sigma_{\mathcal{P}\setminus d'}(d')
}
\end{align}

\vspace{-2mm}\noindent
The score $I_{\mathcal{P}}(d)$ therefore measures the fraction of all cross-paper influence mass in the corpus that is attributed to article $d$.

%% file: figs/dag.tex

\begin{figure}[t]
\centering
\resizebox{0.3\textwidth}{!}{
\begin{tikzpicture}[
    >=Stealth,
    vertex/.style={circle, draw, minimum size=7mm, inner sep=0pt},
    edge/.style={->, thick}
]

\node[vertex] (a) at (0,2) {$a$};
\node[vertex] (b) at (4,2) {$b$};

\node[vertex] (c) at (-1,0) {$c$};
\node[vertex] (d) at (2,0) {$d$};
\node[vertex] (e) at (5,0) {$e$};

\node[vertex] (f) at (1,-2) {$f$};
\node[vertex] (g) at (4,-2) {$g$};

\draw[edge] (a) -- node[midway,left] {$\sigma_a(c)$} (c);
\draw[edge] (a) -- node[midway, right] {$\sigma_a(d)$} (d);
\draw[edge] (b) -- node[midway, right] {$\sigma_b(d)$} (d);
\draw[edge] (b) -- node[midway,right] {$\sigma_b(e)$} (e);

\draw[edge] (a) -- node[midway,right] {$\sigma_a(f)$} (f);
\draw[edge] (c) -- node[midway,left] {$\sigma_c(f)$} (f);
\draw[edge] (d) -- node[midway, right] {$\sigma_d(f)$} (f);
\draw[edge] (d) -- node[midway, right] {$\sigma_d(g)$} (g);
\draw[edge] (e) -- node[midway,right] {$\sigma_e(g)$} (g);

\end{tikzpicture}
}
\vspace{-4mm}
\caption{Contribution propagation on a citation graph.}
\label{fig:toy-citation-graph}
\vspace{-6mm}
\end{figure}
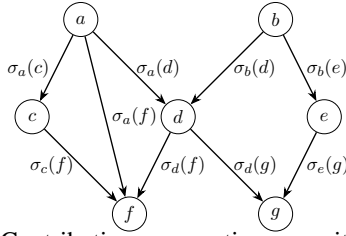

%% file: exp.tex
\section{Experiments}\label{sec:eval}

We evaluate our framework across multiple open-source LLM backbones under both section-level and citation-level evaluation settings. Full implementation details, evaluation setup, metrics, and baseline descriptions are provided in  Appendix~\ref{app:implementation-setup}.

\begin{figure*}[t]
\centering
\includegraphics[width=1\textwidth]{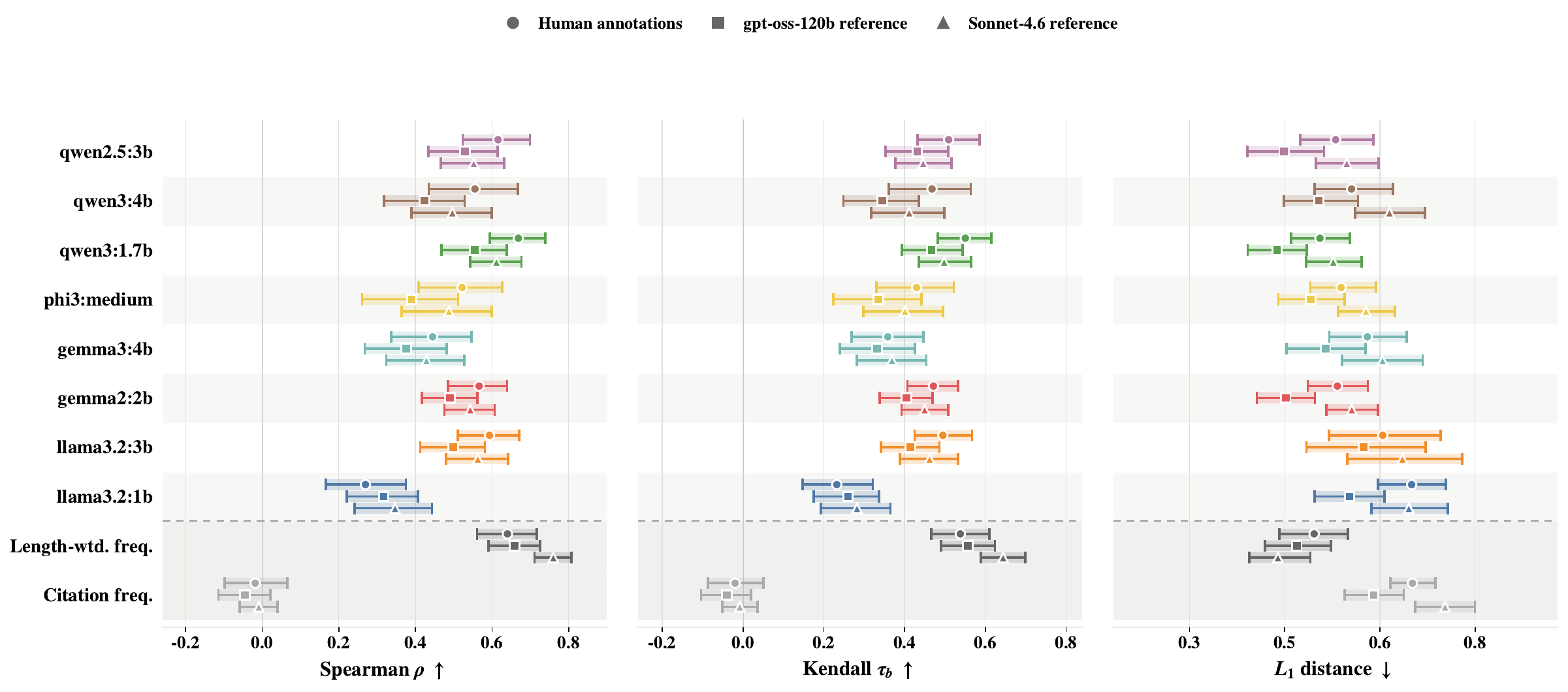}
\vspace{-3mm}
\caption{Evaluation of local LLMs on 
section-level score distributions across 50 papers using human annotations, GPT-OSS-120B, and Claude Sonnet 4.6 as reference sources. Markers show mean agreement across papers, with 95\% bootstrap confidence intervals. $\bullet$, $\blacksquare$, and $\blacktriangle$ denote comparisons with human annotations, GPT-OSS-120B, and Claude Sonnet 4.6, respectively. Higher Spearman's $\rho$ and Kendall's $\tau_b$ indicate stronger rank agreement, while lower $L_1$ distance indicates closer agreement in score allocation. Among the evaluated local LLMs, \texttt{qwen3:1.7b} shows the strongest and most consistent agreement across the three reference evaluators, while the \emph{length-weighted frequency} baseline remains highly competitive, particularly in agreement with the API-based references.}
\label{fig:section_validation}
\vspace{-4mm}
\end{figure*}


\begin{figure*}[t]
\centering
\includegraphics[width=0.9\textwidth]{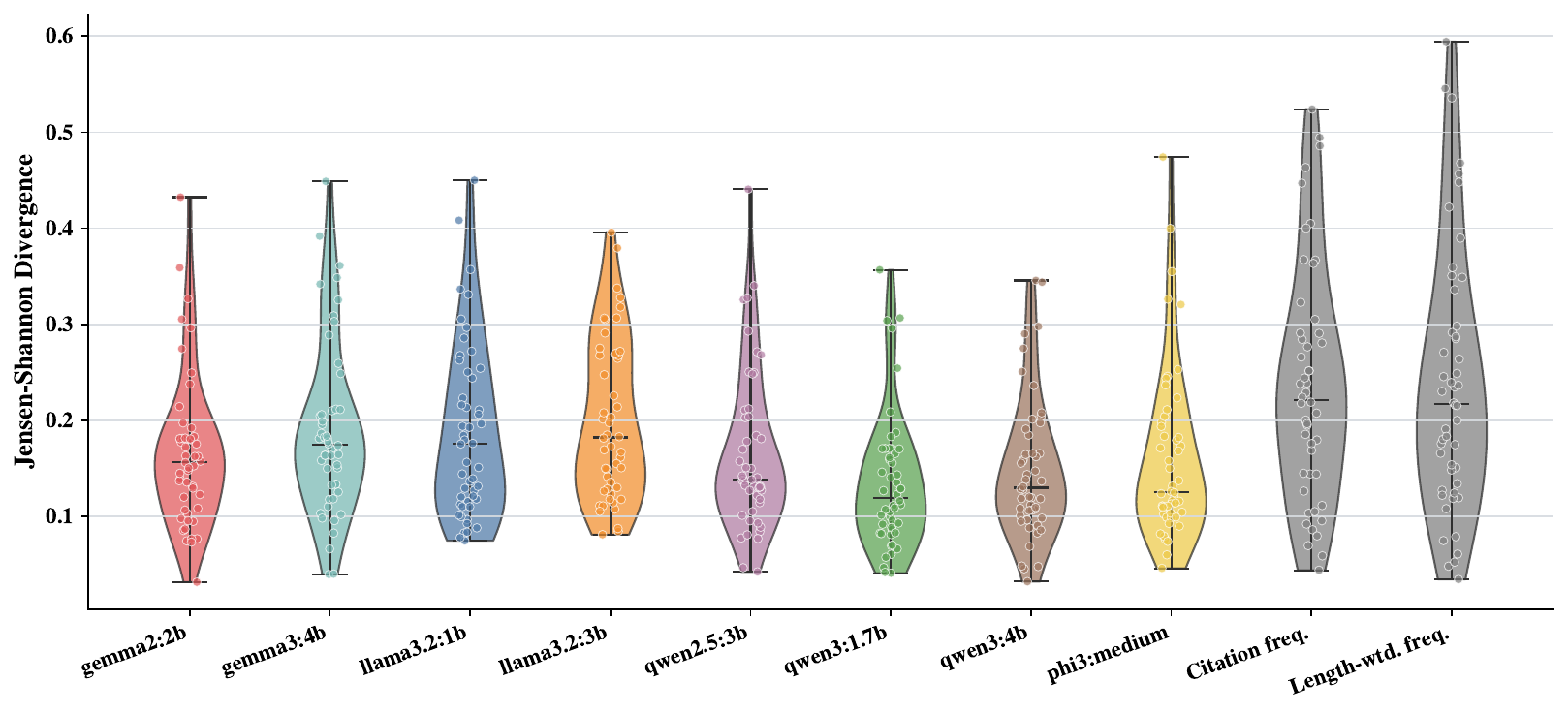}
\vspace{-4mm}
\caption{Distribution of Jensen-Shannon divergence (JSD) between citation contribution scores produced by each model or baseline and the Claude Sonnet 4.6 reference across 50 papers. Lower values indicate closer distributional agreement. \texttt{qwen3:1.7b} exhibits the lowest overall divergence among the local models, while the frequency-based baselines, \emph{citation frequency}, and \emph{length-weighted frequency}, show higher divergence and greater variability.}
\label{fig:citation_jsd_vs_sonnet4.6}
\vspace{-3mm}
\end{figure*}

\begin{figure*}[t]
\centering
\includegraphics[width=0.9\textwidth]{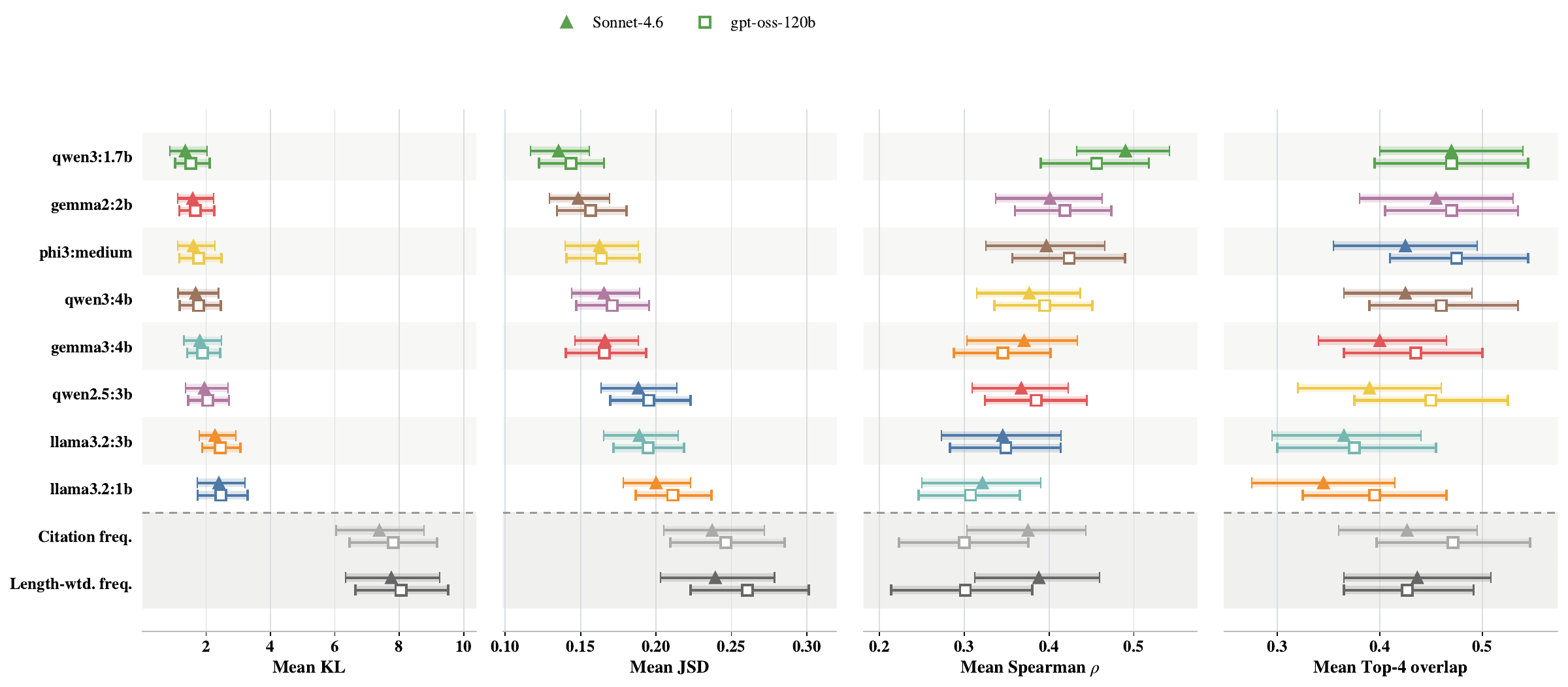}
\vspace{-4mm}
\caption{Mean citation-level agreement with the Claude Sonnet 4.6 and GPT-OSS-120B reference annotations across 50 papers, with 95\% bootstrap confidence intervals. KL divergence and JSD measure distributional disagreement, with lower values indicating better agreement, while Spearman's $\rho$ and Top-4 overlap measure ranking and top-citation agreement, with higher values indicating better performance. Triangles denote Claude Sonnet 4.6 and squares denote GPT-OSS-120B comparisons. The \emph{citation frequency}, and \emph{length-weighted frequency} baselines are shown at the bottom. \texttt{qwen3:1.7b} provides the strongest overall agreement with the Claude Sonnet 4.6 reference and performs best across all four metrics.}

\label{fig:citation_4metrics_2references}
\vspace{-5mm}
\end{figure*}

\subsection{Validation Results}
\label{subsec:validation}

Figures~\ref{fig:section_validation},
\ref{fig:citation_jsd_vs_sonnet4.6}, and
\ref{fig:citation_4metrics_2references} report validation results across the
50 human-annotated papers. We focus on relatively small local LLMs to assess
whether practical and deployable models can serve as contribution-estimation
oracles. In addition to human annotations, we use Claude Sonnet 4.6
~\cite{anthropic2026sonnet46} and GPT-OSS-120B~\cite{gpt-oss-120b} as
complementary reference evaluators. As shown in
Table~\ref{tab:api_human_alignment} and Appendix~\ref{app:human-noise},
both API-based annotators exhibit strong agreement with human judgments.
Using multiple reference evaluators allows us to examine whether the
relative performance of the local models is robust to the choice of annotator.

\paragraph{Section-level validation (Fig.~\ref{fig:section_validation}).} Among the local LLMs, \texttt{qwen3:1.7b} provides the strongest and most consistent section-level estimates across all three reference sources. Against human annotations, it achieves the highest rank agreement among the LLMs, with
Spearman $\rho=0.669$ and Kendall $\tau_b=0.551$, while also obtaining the lowest $L_1$ distance among the LLMs at $0.505$. The same pattern persists against GPT-OSS-120B and Claude Sonnet 4.6, where \texttt{qwen3:1.7b} again performs best among the evaluated LLMs across all three metrics. This
consistency suggests that its advantage is not specific to a single reference source.

The strong performance of the \emph{length-weighted frequency} baseline reveals a different aspect of the section-level task. Against Claude Sonnet 4.6, it achieves the strongest overall agreement across all three metrics, with $\rho=0.760$, $\tau_b=0.644$, and $L_1=0.439$. Against GPT-OSS-120B, it also
achieves the highest rank correlations, although \texttt{qwen3:1.7b} obtains the lowest $L_1$ distance. The pattern changes against human annotations, where \texttt{qwen3:1.7b} achieves the strongest rank agreement and the length-weighted baseline only slightly improves on its $L_1$ distance.
Together, these results indicate that section length captures a substantial structural signal for coarse contribution allocation, particularly for the API-generated references, but does not fully account for human judgments of relative section importance. Content-sensitive estimation remains beneficial
when recovering the human section ranking.

This distinction motivates the citation-level evaluation. If much of section-level importance can be recovered from coarse structural cues, citation attribution provides a more demanding test of whether the estimator captures contribution beyond document structure and citation frequency.

\begin{figure*}[t]
\centering
\includegraphics[width=1\textwidth]{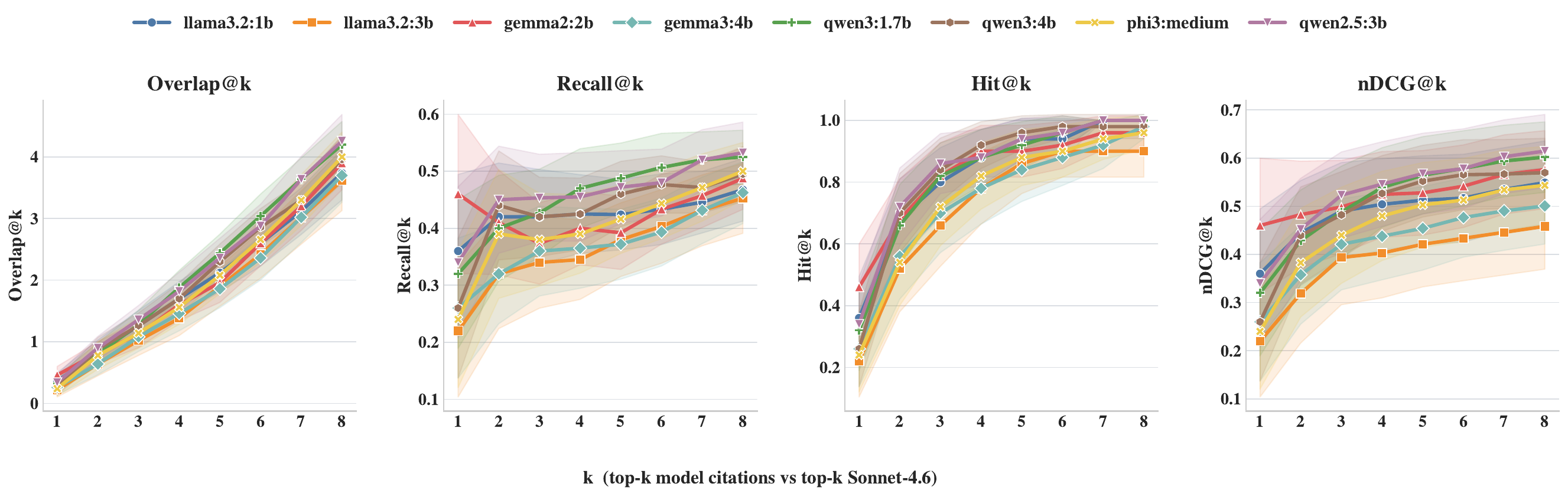}
\vspace{-8mm}
\caption{Top-$k$ sensitivity of citation agreement with the Claude Sonnet 4.6 across $k \in \{1,\ldots,8\}$. For each $k$, the top-$k$ citations produced by each model are compared with the top-$k$ citations in the reference ranking using Overlap@$k$, Recall@$k$, Hit@$k$, and nDCG@$k$. Higher values indicate stronger agreement for all four metrics. Shaded regions denote 95\% bootstrap confidence intervals over the 50 papers. Model differences are most pronounced at small and intermediate values of $k$, while Hit@$k$ approaches saturation as $k$ increases.}
\label{fig:top_k_sensitivity}
\vspace{-5mm}
\end{figure*}


\paragraph{Citation-level validation (Figs.~\ref{fig:citation_jsd_vs_sonnet4.6} and
\ref{fig:citation_4metrics_2references}).}
The citation-level results reveal a substantially different pattern from the section-level evaluation. Figure~\ref{fig:citation_4metrics_2references} compares the complete citation-score distributions, global citation rankings, and agreement among the highest-scoring citations against Claude Sonnet 4.6 and GPT-OSS-120B. Against Claude Sonnet 4.6, \texttt{qwen3:1.7b} achieves the strongest mean agreement across all four metrics, obtaining the lowest KL divergence and JSD together with the highest Spearman correlation and Top-4 overlap. The separation from the frequency-based baselines is particularly pronounced for KL divergence and JSD, indicating that simple citation and length-weighted frequencies do not reproduce the contribution-mass distribution across the cited works.

This contrast with the section-level results is important. Section length
provides a strong prior for coarse contribution allocation, but it is insufficient for determining which individual references contribute most strongly to a paper. The stronger citation-level performance of the LLM-based estimators therefore suggests that contribution attribution requires information beyond section length and citation frequency, including the role of a cited work within the surrounding content.

The aggregate results are also reflected at the individual-paper level.
Figure~\ref{fig:citation_jsd_vs_sonnet4.6} shows that the LLM-based estimators are concentrated at substantially lower JSD values than the two frequency baselines. \texttt{qwen3:1.7b} achieves the lowest mean divergence, while the frequency baselines exhibit higher divergence and heavier upper tails, indicating less consistent agreement with the reference citation distributions across papers. At the same time, the distributions of the stronger local LLMs overlap considerably, suggesting that their differences are smaller than the gap between model-based estimation and the simple frequency heuristics. Paper-level pairwise comparisons provide additional support for this result, with \texttt{qwen3:1.7b} achieving the largest total number of wins across KL divergence, JSD, Spearman's $\rho$, and Top-4 overlap(Appendix~\ref{app:exp:modelsize}).

\begin{table}[t]
\centering
\scriptsize
\setlength{\tabcolsep}{5pt}
\begin{tabular}{lcc}
\toprule
Metric & Sonnet-4.6 & GPT-OSS-120B \\
\midrule
Sp. $\uparrow$     & 0.865 {\scriptsize[$0.828,\,0.900$]} & 0.799 {\scriptsize[$0.750,\,0.844$]} \\
K-$\tau_b$ $\uparrow$ & 0.779 {\scriptsize[$0.728,\,0.826$]} & 0.701 {\scriptsize[$0.644,\,0.757$]} \\
$L_1$ $\downarrow$    & 0.320 {\scriptsize[$0.280,\,0.363$]} & 0.348 {\scriptsize[$0.302,\,0.399$]} \\
Ov.@4 $\uparrow$   & 1.720 & 1.320 \\
Rec.@4 $\uparrow$  & 0.430 & 0.330 \\
Hit@4 $\uparrow$   & 0.900 & 0.840 \\
nDCG@4 $\uparrow$  & 0.481 & 0.375 \\
\bottomrule
\end{tabular}
\vspace{-2mm}
\caption{Alignment of API-generated scores with human annotations over 50 papers. The first three rows report mean section-level agreement with 95\% bootstrap confidence intervals; the last four rows report citation top-4 agreement.}
\label{tab:api_human_alignment}
\vspace{-4mm}
\end{table}

\input{ablation}

%% file: ablation.tex
\subsection{Sensitivity and Ablation Analysis}\label{ablation}

\paragraph{Effect of $k$.}
Figure~\ref{fig:top_k_sensitivity} examines the sensitivity of citation-level agreement to the retrieval cutoff over $k\in\{1,\ldots,8\}$. We use $k=4$ as the primary evaluation point because the human annotation task explicitly asks experts to identify the four citations they consider most foundational to the paper. Evaluating a broader range of $k$ allows us to determine whether the observed model differences depend strongly on this choice.

Model behavior varies most clearly at small and intermediate values of $k$. At $k=4$, \texttt{qwen3:1.7b} achieves the highest Overlap@$4$ and Recall@$4$, while \texttt{qwen2.5:3b} obtains the highest nDCG@$4$ and \texttt{qwen3:4b} achieves the highest Hit@$4$. As $k$ increases, Hit@$k$ rises rapidly and approaches saturation for most models, making it progressively less discriminative. Recall@$k$ and nDCG@$k$ continue to
separate models at larger cutoffs and therefore provide a more informative view of how much of the high-contribution citation set is recovered and how well that set is ranked. Overlap@$k$ also increases with $k$ by construction, so its values are most informative when comparing models at the same cutoff.

\paragraph{Effect of model size.}
Parameter count does not exhibit a consistent relationship with contribution estimation quality. This pattern is particularly clear when the section-level and citation-level evaluations are considered together. For \texttt{llama3.2}, increasing model size from 1B to 3B substantially improves
section-level agreement, yet the 1B variant performs better across the citation Top-$k$ metrics in Figure~\ref{fig:top_k_sensitivity}. For the Gemma models, \texttt{gemma2:2b} consistently outperforms \texttt{gemma3:4b} in the citation-level sensitivity analysis, while \texttt{qwen3:1.7b} generally
outperforms \texttt{qwen3:4b} on Overlap@$k$, Recall@$k$, and nDCG@$k$ despite having fewer than half as many parameters. These results show that {\em larger models are not necessarily better contribution estimators} and that relative performance can depend on whether the task concerns coarse section-level allocation or fine-grained citation attribution. See Appendix~\ref{app:exp:modelsize} for a detailed pairwise analysis showing that larger models do not consistently outperform smaller ones across the evaluated metrics.

\input{citation-network}

\input{score_definitions_table}

%% file: citation-network.tex
\subsection{Corpus-level Case Study}
\label{sec:citation-network}
The weighted citation graph $G_{\mathcal{P}}$ produces two complementary quantities for each paper $d \in \mathcal{P}$, the corpus-level contribution score $\sigma_{\mathcal{P}}(d)$
(Equation~\ref{eq:corpus_contribution}) and the normalized cross-paper influence score $I_{\mathcal{P}}(d)$ (Equation~\ref{eq:norm_influence}). Unlike raw citation counts, which treat all incoming citations equally, the proposed framework weights citation relationships by their estimated contribution and recursively propagates credit across citation paths.

To evaluate the framework on a downstream corpus-level task, we construct an
induced citation graph consisting of 17 mutually connected papers and apply the full scoring pipeline using \texttt{qwen3:1.7b} as the backbone model. The resulting scores are reported in Table~\ref{tab:influence_scores}, and the corresponding weighted citation graph is shown in Fig.~\ref{fig:citation_network_Ip} in Appendix~\ref{app:corpus-case}.

\begin{table}[t]
\centering
\resizebox{.5\textwidth}{!}{
\begin{tabular}{@{}l@{}c@{}r@{\hspace{2mm}}r@{}}
\toprule
\textbf{Paper} & \textbf{Indg} & {$\boldsymbol{I_\mathcal{P}(d)}$} & {$\boldsymbol{\sigma_{\mathcal{P}}(d)}$} \\
\midrule
N-Machine-Translation\cite{n-machine-trans}             & 3          & \textbf{4.75e-1} & 8.34e-1 \\
AttentionAllYouNeed\cite{vaswani2017attention}          & \textbf{6} & 1.22e-1          & 6.54e-1 \\
shift-reduce-const-pars\cite{shift-reduce}              & 1          & 9.16e-2          & 7.28e-1 \\
LSM-Memory\cite{lsm-memo}                               & 1          & 8.76e-2          & 7.61e-1 \\
struc-attention-net\cite{struct-attention}              & 1          & 5.80e-2          & 7.35e-1 \\
output-embed-llm\cite{output-embed}                     & 1          & 5.67e-2          & 7.57e-1 \\
yolov10\cite{yolov10}                                   & 1          & 3.99e-2          & 7.37e-1 \\
gemma2\cite{gemma2}                                     & 2          & 2.31e-2          & 8.41e-1 \\
Gen-Sequences-RNN\cite{gen-sequ-rnn}                    & 2          & 1.86e-2          & 7.15e-1 \\
Deep-Residual-Learning\cite{deep-redi-learning}         & 2          & 1.15e-2          & 6.31e-1 \\
llama3\cite{llama3}                                     & 1          & 9.25e-3          & 8.42e-1 \\
ExploringLimitLLMs\cite{exploring}                      & 1          & 6.96e-3          & 7.55e-1 \\
Smarter-Better-Faster-Longer\cite{smarter-better}       & 0          & 0                & 6.82e-1 \\
kimiK2\cite{kimik2}                                     & 0          & 0                & \textbf{8.54e-1} \\
YOLO-World-Real-Time\cite{yolo-world}                   & 0          & 0                & 8.04e-1 \\
d-fine-redefine\cite{d-fine}                            & 0          & 0                & 7.86e-1 \\
simPo\cite{simpo}                                       & 0          & 0                & 8.33e-1 \\
\bottomrule
\end{tabular}
}
\vspace{-1mm}
\caption{Normalized cross-paper influence score $I_{\mathcal{P}}(d)$ and corpus-level contribution
$\sigma_{\mathcal{P}}(d)$ for each paper. $I_{\mathcal{P}}(d)$ reflects propagated cross-paper influence
through direct and indirect citations. 
$\sigma_{\mathcal{P}}(d)$ captures the total corpus-level contribution of each paper, combining its intrinsic technical score with influence received
via citation paths.}
\label{tab:influence_scores}
\vspace{-3mm}
\end{table}

\vspace{-2mm}
\paragraph{Corpus-level contribution.}
The two scores capture substantially different properties of a paper within the induced corpus. \textit{kimiK2}~\cite{kimik2} achieves the highest corpus-level contribution score, with
$\sigma_{\mathcal{P}}=0.854$, despite having $I_{\mathcal{P}}=0$. Similarly, \textit{simPo}~\cite{simpo} and \textit{d-fine-redefine}~\cite{d-fine} retain relatively high $\sigma_{\mathcal{P}}$ scores while receiving no propagated cross-paper influence. Under our edge convention, where edges point from a citing paper to a cited paper, these are zero-indegree nodes in the induced graph and
therefore receive no propagated credit from other papers in $\mathcal{P}$. Their corpus-level scores consequently reflect their intrinsic technical contribution rather than inherited influence.
\paragraph{Cross-paper influence.}
The influence ranking differs substantially from both corpus-level contribution and raw in-degree. \textit{N-Machine-Translation}~\cite{n-machine-trans} has the largest normalized influence score, $I_{\mathcal{P}}=0.475$, despite receiving only three citations within this corpus. In contrast, \textit{Attention Is All You Need}
~\cite{vaswani2017attention} has the highest in-degree of six but a substantially smaller influence score of $I_{\mathcal{P}}=0.122$. The two papers account for approximately $59.7\%$ of the propagated influence mass. This contrast illustrates that citation frequency
alone does not determine propagated influence. The identity and influence of the citing papers, the learned edge weights, and the paper's position along indirect citation paths all affect $I_{\mathcal{P}}$.

The same effect appears among papers with identical in-degree. \textit{yolov10}~\cite{yolov10}, \textit{LSM-Memory}~\cite{lsm-memo}, and
\textit{Exploring-Limit-LLMs}~\cite{exploring} each receive one in corpus citation, yet their influence scores differ substantially. More broadly, the three highest values in Table~\ref{tab:influence_scores} belong to different papers across the three quantities. \textit{kimiK2} has the largest $\sigma_{\mathcal{P}}$, \textit{N-Machine-Translation} has the largest $I_{\mathcal{P}}$, and \textit{Attention Is All You Need} has the largest in-degree. The framework therefore distinguishes intrinsic corpus-level contribution, propagated structural influence, and citation frequency rather than collapsing them into a single notion of impact.

Finally, these scores should be interpreted as measurements within the selected induced graph. A paper may have zero propagated influence because no paper in the selected corpus cites it, even if it is influential outside the corpus. Likewise, both $\sigma_{\mathcal{P}}$ and
$I_{\mathcal{P}}$ depend on which papers and citation paths are included. The resulting rankings therefore characterize local structural relationships within $\mathcal{P}$ rather than global scientific impact.

%% file: conclusion.tex
\section{Conclusion and Future Work}

We introduced contribution-based credit scoring, a framework that decomposes a paper's credit into original and citation-derived contribution and propagates it through weighted citation graphs. Our results suggest that LLMs can provide useful local contribution estimates beyond citation frequency and document-structure heuristics.

An immediate future direction is to develop models specialized for contribution scoring. The strong performance of small local models suggests that training alignment may matter more than scale, motivating instruction-tuning compact LLMs for contribution estimation in scientific articles. Another direction is to apply our framework to large research corpora, where contribution-weighted citation graphs may reveal domain-level influence signals beyond citation counts.

%% file: limitation.tex
\section*{Limitations}

Our framework inherits several limitations from both the annotation process and the structure of scientific writing itself.

The first challenge is the absence of a fully objective ground truth for contribution attribution. Even expert researchers may disagree on which citations are most central to a paper or on how importance should be distributed across sections. Our annotation study reflects this ambiguity directly: after examining model outputs, annotators frequently revised their original top-4 selections. The reported evaluation scores should therefore be interpreted as agreement with a particular set of expert judgments rather than exact measurements of scientific contribution.

The evaluation benchmark is also relatively small. Our experiments cover 50 annotated papers and a separate 17-paper citation graph for the corpus-level analysis. While this is sufficient to study model behavior and validate the framework qualitatively, it does not fully represent the scale and diversity of universal scientific citation networks.

The framework also depends heavily on the underlying language model. Our experiments show that model scale alone is not a reliable predictor of performance, and different model families exhibit noticeably different ranking behavior. The resulting contribution scores may therefore change with the choice of backbone model, prompt design, quantization level, or sampling strategy. Repeated-query averaging reduces variance but does not remove model-specific biases.

Another limitation comes from the document structure assumption itself. The framework distributes contribution through the section hierarchy and local sibling comparisons, which works well for standard scientific papers but may fail when important ideas are concentrated in short sections, figures, appendices, or scattered discussions. In some cases, section length also acts as an unintended signal, since both annotators and models often associate longer methodological sections with higher importance.

The corpus-level propagation analysis should also be interpreted carefully. The propagated influence scores depend strongly on the selected citation subgraph. A paper may receive little propagated influence simply because the surrounding corpus is incomplete or because descendant works are absent from the induced graph. Unless the corpus approaches the scale and coverage of the global scholarly citation network, the resulting rankings should be interpreted
as measurements of structural influence within the selected corpus rather
than estimates of universal scientific impact.


%% file: appendix/Toc.tex
\subsection*{Table of Content}
\begin{enumerate}[leftmargin=*]
    \item Related Work \dotfill \ref{sec:related}
    \item Modeling as a Cooperative Game\dotfill \ref{sec:warmup}
    \item Proof of Theorem~\ref{th:axiom} \dotfill \ref{app:axioms}
    \item LLM-as-an-Estimator\dotfill\ref{app:llm}
    \item Proof of Theorem~\ref{th:1}\dotfill\ref{app:proof}
    \item Efficient Computation of Corpus-level Contribution Scores \dotfill \ref{app:topological}
    \item Implementation Details\dotfill \ref{app:implementation-setup}
    \item Extended Experiment Details\dotfill\ref{app:exp:extended}
    \item LLM Prompts\dotfill \ref{app:prompts}
\end{enumerate}

%% file: appendix/related.tex
\section{Related Work}\label{sec:related}

Prior work on citation analysis can be broadly divided into three categories: (1) frequency-based bibliometric methods, (2) content-based citation analysis, and (3) computational approaches based on natural language processing.

Traditional bibliometric approaches treat citations as uniform signals of impact, relying primarily on citation counts and h-index. While effective for large-scale analysis, these methods ignore the textual context and functional role of citations, treating all references as equally informative.

Content-based citation analysis (CCA) extends this paradigm by incorporating the syntactic and semantic context of citations. This line of work studies where citations appear within a paper (e.g., introduction versus methodology) and why they are used (e.g., background, comparison, or extension). Subsequent work in natural language processing has further operationalized these ideas through tasks such as citation function classification, sentiment analysis, citation summarization, and citation-based retrieval and recommendation~\cite{cca}. These approaches provide a richer understanding of citation usage, but they remain primarily qualitative or categorical, focusing on labeling or interpreting individual citation contexts.

Recursive citation-ranking methods instead model scientific influence through propagation over citation graphs, where citations from influential papers contribute more strongly to the score of cited work \cite{senanayake2015pagerank, ma2008bringing, west2010eigenfactor}. These approaches capture structural dependence and long-range influence patterns across citation paths, extending beyond raw citation frequency alone. The resulting scores, however, are typically derived from graph topology or citation frequency and do not account for the semantic role or contribution-level importance of citations within the citing document \cite{walker2007ranking, valenzuela2015identifying}. Work on influential citation detection partially addresses this limitation by identifying references that are central to a paper rather than merely co-mentioned \cite{cca, Cohan2019Structural}. Our framework differs from these approaches by combining hierarchical contribution attribution inside documents with semantically weighted propagation across citation graphs.

Recent work has also explored using large language models (LLMs) to directly evaluate research quality by assigning a single scalar score to a paper based on its title and abstract \cite{paper-eval}. These methods show moderate correlation with expert judgments, suggesting that LLMs can capture high-level signals of quality. However, they treat evaluation as a black-box mapping from surface text to a holistic score, implicitly conflating dimensions such as novelty, rigor, and significance without explicitly modeling them. Moreover, because they operate at the document level and primarily rely on abstracts, they do not capture how contributions are distributed across different sections of a paper or how individual components contribute to the overall value. As a result, they lack interpretability and do not provide a mechanism for attributing credit within the document.

Another key limitation of existing approaches is the lack of a unified, quantitative framework for attributing contribution across the entire structure of a document. Prior methods typically operate at a local level (e.g., sentence or citation context) and treat units independently, without modeling how different parts of a paper jointly contribute to its overall value. Moreover, they do not explicitly distinguish between the contribution of original content and that derived from cited work, nor do they enforce consistency constraints across different levels of granularity.

In contrast, our work formulates citation analysis as a hierarchical importance-allocation problem over the full document structure. We represent a paper as a rooted tree and assign normalized importance scores that are conserved across levels, ensuring that the total importance of the paper is distributed among its constituent sections, paragraphs, and citations. This introduces a globally consistent view of contribution in which all components are coupled through conservation constraints.

Our framework further decomposes paragraph-level importance into original and citation-derived components, enabling fine-grained attribution of credit to both original contributions and referenced prior work. Aggregating original contributions across the hierarchy produces a normalized original-contribution score for the full document, while distributing citation-derived importance across referenced works yields a normalized citation-contribution score for each reference relative to the paper as a whole.

Overall, our approach shifts citation analysis from local qualitative interpretation and black-box document-level scoring toward a global, quantitative, and structurally grounded framework for contribution attribution in scientific documents.

%% file: appendix/warmup.tex
\section{Modeling as a Cooperative Game}\label{sec:warmup}

In this section, we formalize the contribution-based credit scoring of the research articles as a {\em cooperative game}.
A scientific document does not emerge in isolation; rather, it is the result 
of a combination of original contributions and the integration of prior work 
through citations. We model this interaction as a cooperative setting in which 
multiple entities jointly contribute to the final outcome.

Let $d$ denote a target document, and let $\mathcal{C}(d)$ be the set of all 
references cited by $d$. We define the set of players as
\begin{equation}\label{eq:contribution}
N = \{d\} \cup \mathcal{C}(d),
\end{equation}
where 
\begin{enumerate}[leftmargin=*]
    \item {\bf original contribution}: the document itself represents its original contribution, and
    \item {\bf citation contribution}: each reference $c \in \mathcal{C}(d)$ represents a cited work that contributes to $d$.
\end{enumerate}

We consider a transferable utility (TU) cooperative game~\cite{myerson2013game} 
defined by the pair $(N, \nu)$, where
\(
\nu : 2^N \rightarrow \mathbb{R}^{+}
\)
is a characteristic function. For any coalition $S \subseteq N$, the value 
$\nu(S)$ represents the total contribution that can be generated by the subset 
of players in $S$. Intuitively, $\nu(S)$ captures the extent to which the ideas, 
methods, or context provided by the elements in $S$ are sufficient to produce 
a version of the document's outcome.

In particular, the grand coalition $N$ corresponds to the complete document, 
and we normalize its value as
\(
\nu(N)=1,
\)
representing the total contribution of the document. The goal is therefore to 
allocate this total value among the players in $N$ in a manner that reflects 
their respective contributions.

A standard solution concept for such attribution problems is the 
{\bf Shapley value}~\cite{shapley1953value}. For each player $i\in N$, the 
Shapley value $\phi_i$ is defined as

{\small
\begin{align*}
\phi_i
=
\sum_{S\subseteq N\setminus\{i\}}
\frac{|S|!(|N|-|S|-1)!}{|N|!}
\Big(
\nu(S\cup\{i\})-\nu(S)
\Big).
\end{align*}
}

This quantity represents the expected marginal contribution of player $i$ 
over all possible permutations of the coalition-formation process.

The Shapley value satisfies several desirable properties, including 
efficiency, dummy, and additivity, making it a principled 
mechanism for fair credit allocation. As a result, it provides a natural way to quantify how much each citation and the document's own original content contribute to the overall value of the work.

Despite its conceptual appeal, the Shapley value formulation is not directly 
applicable in our setting. Its computation requires access to the 
characteristic function $\nu(S)$ for all coalitions $S\subseteq N$, which 
corresponds to evaluating counterfactual versions of the document 
constructed from arbitrary subsets of references. Such counterfactuals are 
not observable: a scientific document is produced once, with a fixed set of 
references, and the process cannot be rerun under alternative coalitions. 
Consequently, unlike standard cooperative games, we do not have access to a 
value oracle that can evaluate $\nu(S)$.

Furthermore, this setting lacks the structural properties that enable 
practical credit assignment in related domains such as reinforcement 
learning~\cite{foerster2018counterfactual,li2021shapley}. There is no repeated 
interaction, no controllable environment, and no well-defined reward signal 
that can be sampled across different trajectories or team compositions. 
Combined with the exponential complexity of Shapley value 
computation~\cite{chen2023algorithms}, these limitations render direct 
application infeasible.

This motivates our shift from counterfactual coalition evaluation to a 
{\em structural, observation-based attribution framework} grounded directly 
in the document itself. Our approach is based on the notion of a 
{\em contribution tree}, which we introduce next.

%% file: appendix/axioms.tex
\section{Proof of Theorem~\ref{th:axiom}}\label{app:axioms}

\noindent
\textbf{Theorem~\ref{th:axiom}.} 
    \emph{Contribution Tree satisfies {\em efficiency}, {\em additivity}, and {\em hierarchical conservation}}.

\begin{proof}
We prove each property in turn.

\paragraph{Hierarchical conservation.}
For every internal node $u$, the outgoing edge weights form a distribution:
\[
\sum_{v\in \mathrm{ch}(u)} w(u,v)=1.
\]
Since each child receives importance
\[
\sigma(v)=w(u,v)\sigma(u),
\]
we have
    \begin{align*}
        \sum_{v\in \mathrm{ch}(u)}w(u,v)\, \sigma(u)
        &=  \sigma(u)\sum_{v\in \mathrm{ch}(u)}w(u,v) \\
        &= \sigma(u)
    \end{align*}
Thus, the importance of every internal node is exactly distributed among its children.

\paragraph{Efficiency.}
We show that the total importance assigned to the leaves is equal to the importance of the root. Starting from the root, whose importance is normalized as $\sigma(d)=1$, hierarchical conservation implies that replacing any internal node by its children preserves the total importance mass. Repeatedly applying this replacement until all internal nodes have been expanded to leaves therefore preserves the total mass. Hence,
\[
\sum_{\ell\in \mathrm{leaves}(T)} \sigma(\ell)=\sigma(d)=1.
\]
Thus, the full contribution of the article is distributed across the leaves.

\paragraph{Additivity.}
Let $S\subseteq V$ be any set of nodes whose aggregate contribution is defined as the sum of their individual importance scores:
\[
\sigma(S)=\sum_{u\in S}\sigma(u).
\]
For two disjoint sets of nodes $S_1,S_2\subseteq V$, we have

\begin{align*}
\sigma(S_1\cup S_2)
&=
\sum_{u\in S_1\cup S_2}\sigma(u)\\
&=
\sum_{u\in S_1}\sigma(u)+\sum_{u\in S_2}\sigma(u)\\
&=
\sigma(S_1)+\sigma(S_2).
\end{align*}

More generally, for any finite collection of pairwise disjoint node sets $\{S_i\}_{i=1}^m$,
\[
\sigma\!\left(\bigcup_{i=1}^m S_i\right)
=
\sum_{i=1}^m \sigma(S_i).
\]
Therefore, contribution is additive across disjoint components of the tree.
\end{proof}

%% file: appendix/llm-estimator.tex
\section{LLM-as-an-Estimator}
\label{app:llm}

Section~\ref{sec:weights} defines contribution estimation through an abstract oracle: given a parent node $u$ and its children $\mathrm{ch}(u)$, the oracle returns non-negative sibling-relative scores $X(v \mid u)$, which are normalized to obtain edge-weight estimates $\widehat{w}(u,v)$. In principle, such scores could be elicited from human experts who have read and understood the paper. However, expert annotation is costly, slow, and impractical at scale, especially because contribution judgments must be made repeatedly across document units and citation contexts. This motivates an automated instantiation of the contribution-estimation oracle.

We instantiate the oracle using large language models (LLMs). Recent work on LLM-as-a-judge evaluation suggests that sufficiently capable LLMs can approximate human preferences and quality judgments in open-ended evaluation tasks~\cite{gu2024survey}, often achieving substantial agreement with human evaluators, while also exhibiting biases and variance that must be explicitly controlled~\citep{zheng2023judging,liu2023geval,chen2024humans}. This line of work has also begun to extend to scientific assessment and peer-review assistance, where LLMs have been studied as tools for structured paper evaluation, review generation, and meta-review support~\cite{deepreview,metareviewer}. More broadly, this progress has motivated major computer science conferences, including NeurIPS 2026, to explore opt-in AI-assisted reviewing workflows.

Accordingly, we do not treat the LLM as a source of ground-truth contribution scores or as a substitute for expert scientific judgment. Instead, we use it as a noisy comparative estimator of local importance, consistent with the noise model in Equation~\ref{eq:estimation_with_noise}. This distinction is important: our framework does not ask the model to decide the absolute quality, novelty, or significance of a paper. Rather, it asks the model to perform a more constrained relative-comparison task within the observed document structure.

For each internal node $u$, the model is given the content associated with $u$ together with the child units $v \in \mathrm{ch}(u)$. It then returns non-negative scores \(X^{(t)}(v \mid u)\), intended to reflect the relative contribution of each child within the local context of $u$. The LLM is therefore not asked to assign a global paper-level quality or credit score, nor to evaluate citations independently of the document structure. Rather, it performs a sequence of localized sibling-comparison tasks, each asking how the importance mass of a parent unit should be redistributed among its children. The resulting raw scores are normalized across siblings using Equation~\ref{eq:sample-normalization}, and the response set $\mathcal{A}(u)$ is then aggregated using Equation~\ref{eq:aggregate-estimator}.

This design also mitigates, though does not eliminate, known limitations of LLM-based evaluation. Since only sibling-relative scores are used, arbitrary differences in the scale of model outputs across different parent nodes are removed by local normalization. Since multiple samples are collected for each parent node, response variance can be reduced by averaging over valid normalized allocations. Finally, because the contribution tree enforces non-negativity and conservation constraints, the final scores remain structurally consistent even when individual model responses are noisy.

%% file: appendix/proof.tex
\section{Proof of Theorem~\ref{th:1}}\label{app:proof}

\noindent
\textbf{Theorem~\ref{th:1}.} 
\emph{
    Let \(\mathcal C_{\mathcal P}(d)=\mathcal C(d)\cap \mathcal P\) denote the citations of \(d\) that belong to the corpus. Suppose that, after restricting the citation graph to \(\mathcal P\), each paper satisfies \(
    \sigma_{\mathrm{orig}}(d)
    +
    \sum_{c\in \mathcal C_{\mathcal P}(d)} \sigma_d(c)
    =
    1
    \). Then, \(
        \sum_{d\in\mathcal P}\sigma_{\mathcal P}(d)=|\mathcal P|.
    \)
}

\begin{proof}
For each paper \(d\in\mathcal P\), define the total mass arriving at \(d\) by
\[
    M(d)
    =
    1+
    \sum_{\pi\in\{(a\rightsquigarrow d)\}_{\mathcal P}}
    \prod_{(u\to v)\in\pi}\sigma_u(v),
\]
where the leading \(1\) is the unit mass initially assigned to \(d\), and the
summation is over all nonempty directed paths in the corpus-level citation graph
that end at \(d\). Following Equation~\ref{eq:corpus_contribution},
\[
    \sigma_{\mathcal P}(d)
    =
    \sigma_{\mathrm{orig}}(d)M(d).
\]

We show that contribution propagation conserves the total mass. Since the citation
graph is a DAG, fix a topological ordering of its nodes. When a paper \(d\) is
processed in this ordering, all mass that can reach \(d\) through directed paths
has already been accumulated in \(M(d)\). The paper retains
\(
    \sigma_{\mathrm{orig}}(d)M(d)
\)
as its own corpus-level original contribution, and propagates to each
in-corpus citation \(c\in\mathcal C_{\mathcal P}(d)\) the amount
\(
    \sigma_d(c)M(d).
\)
Therefore, the total mass accounted for at \(d\), including both retained and
propagated mass, is
\[
    \sigma_{\mathrm{orig}}(d)M(d)
    +
    \sum_{c\in\mathcal C_{\mathcal P}(d)} \sigma_d(c)M(d).
\]
By the theorem assumption,
\[
    \sigma_{\mathrm{orig}}(d)
    +
    \sum_{c\in\mathcal C_{\mathcal P}(d)} \sigma_d(c)
    =
    1.
\]
Hence,
\[
    \sigma_{\mathrm{orig}}(d)M(d)
    +
    \sum_{c\in\mathcal C_{\mathcal P}(d)} \sigma_d(c)M(d)
    =
    M(d).
\]
Thus, processing \(d\) does not generate any additional mass; it only partitions
the incoming mass \(M(d)\) into a retained part and a propagated part.

The only mass introduced into the system is the initial unit mass assigned to
each paper in \(\mathcal P\). Hence the total injected mass is
\[
    \sum_{d\in\mathcal P}1 = |\mathcal P|.
\]
Because the graph is finite and acyclic, every propagated unit of mass follows a
finite directed path. If a paper \(d\) has no outgoing in-corpus citation edges,
then \(\mathcal C_{\mathcal P}(d)=\emptyset\), and the normalization condition
implies \(\sigma_{\mathrm{orig}}(d)=1\). Thus, all mass arriving at such a sink paper is retained. Consequently, {\em no mass is lost, and no additional mass is created during propagation}.
As a result, the sum of all retained corpus-level contributions equals the
total injected mass:
\[
    \sum_{d\in\mathcal P}\sigma_{\mathcal P}(d)
    =
    \sum_{d\in\mathcal P}\sigma_{\mathrm{orig}}(d)M(d)
    =
    |\mathcal P|.
\]
\end{proof}

%% file: appendix/topological.tex
\section{Efficient Computation of Corpus-level Contribution Scores}\label{app:topological}
Let \(\{(a\warrow d)\}_{\mathcal{P}}\) denote the set of directed paths in \(G_{\mathcal{P}}\).
The number of paths in \(\{(a\warrow d)\}_{\mathcal{P}}\) can be as large\footnote{$n=\vert\mathcal{P}\vert$. A chain with \(n\) nodes has \(\frac{n(n-1)}{2}\) paths.} as \(O(n^2)\), while each path can have length as large as \(O(n)\). As a result, a brute-force approach for 
computing the corpus-level contribution scores using Equation~\ref{eq:corpus_contribution} would take \(O(n^3)\). Therefore, in this section, we present an efficient algorithm with a time complexity linear to the number of edges in $G_{\mathcal{P}}$.

At a high level, instead of directly applying Equation~\ref{eq:corpus_contribution}, for each article \(d\in \mathcal{P}\), we first collect the total mass of contributions from its parents who cite this article, and only when the mass from all of its parents is collected, we propagate it to its children.

Specifically, we follow the topological ordering of the nodes in \(G_\mathcal{P}\) for mass collection and propagation. A topological ordering is a linear ordering of the nodes such that for every directed edge \((u\rightarrow v)\in E\), node \(u\) appears before node \(v\) in the ordering~\cite{kleinberg2006algorithm}. 
As a result, processing the nodes $p$ according to the topological ordering guarantees that all of its parents are processed before it; hence the total mass from all its parents are collected, ready to be propagated to its children.

\input{alg}

Algorithm~\ref{alg:topological-ordering} adapts the topological-ordering algorithm~\cite{kleinberg2006algorithm} by maintaining two quantities for each article \(d\): its current indegree \(\mathrm{indeg}(d)\), which counts the number of unprocessed parents of \(d\), and its accumulated contribution mass \(\mathrm{mass}(d)\).

Initially,
\(\mathrm{mass}(d)=1\) for every article \(d\), corresponding to the unit credit
assigned to the article itself. As the algorithm processes earlier nodes in the
topological order, additional credit mass is propagated to \(d\) through incoming
citation paths.

When an article \(d\) is removed from the zero-indegree set, all credit that can
reach \(d\) through previously processed paths has already been accumulated in
\(\mathrm{mass}(d)\). The corpus-level original contribution of \(d\) is then
computed as \(\sigma_{\mathcal{P}}(d) = \sigma_{\mathrm{orig}}(d)\,\mathrm{mass}(d).\)
The remaining mass is propagated to each cited article \(v\) through the weighted edge \((d\rightarrow v)\). Specifically, article \(v\) receives an additional amount \(\sigma_d(v)\,\mathrm{mass}(d)\), which corresponds to the fraction of \(p\)'s accumulated credit attributed to
reference \(v\). After this propagation step, the indegree of \(v\) is decremented. Once all parents of \(v\) have been processed, \(v\) is added to the zero-indegree set and becomes ready for processing.
This procedure is equivalent to summing, for every article \(d\), the contributions of all directed paths ending at \(d\), as in Equation~\ref{eq:corpus_contribution}. However, it avoids explicit path enumeration by aggregating all path contributions locally at each node before propagating them forward. 

\begin{lemma}
    The time complexity of Algorithm~\ref{alg:topological-ordering} is \(O(|\mathcal{P}|+|E|)\)
\end{lemma}

\begin{proof}
    Initializing the indegree and mass variables for each node in $G_{\mathcal{P}}$ takes $O(|\mathcal{P}|)$. After that, following the topological order, the algorithm visits each node once, while visiting each edge once to update the indegree variables. As a result, the time complexity of the algorithm is \(O(|\mathcal{P}|+|E|)\).
\end{proof}


%% file: alg.tex
\begin{algorithm}[t]
\caption{Corpus-Contribution Propagation}
\label{alg:topological-ordering}
\begin{algorithmic}[1]
\Require The DAG \(G_{\mathcal{P}}=(\mathcal{P},E)\)
\Ensure The corpus-level original contributions (Equation~\ref{eq:corpus_contribution})

    \For{\(d\in \mathcal{P}\)}
        \State \(\mathrm{indeg}(v) \gets 0;~\mathrm{mass}(v)\gets 1\) 
    \EndFor
    \For{\((u\rightarrow v)\in E\)}
        \State \(\mathrm{indeg}(v) \gets \mathrm{indeg}(v)+1\)
    \EndFor
    
    \State \(S \gets \{v\in V : \mathrm{indeg}(v)=0\}\)

    \While{\(S\neq \emptyset\)}
        \State Remove an arbitrary node \(d\) from \(S\)
        \State \(\sigma_{\mathcal{P}}(d) \gets \sigma_\mathrm{orig}(d)\, \mathrm{mass}(d)\)
        \For{\((d\rightarrow v)\in E\)}
            \State $\mathrm{mass}(v)\gets \mathrm{mass}(v)+\sigma_d(v)\mathrm{mass}(d)$
            \State \(\mathrm{indeg}(v) \gets \mathrm{indeg}(v)-1\)
            \If{\(\mathrm{indeg}(v)=0\)}
                \State Add \(v\) to \(S\)
            \EndIf
        \EndFor
    \EndWhile

    \State {\bf Return} \(\{\sigma_\mathcal{P}(d)\}_{d\in\mathcal{P}}\)

\end{algorithmic}
\end{algorithm}

%% file: appendix/exp.tex
\section{Implementation Details}
\label{app:implementation-setup}
We instantiate the framework using an LLM in a repeated-query setting. For each parent node in the section hierarchy, the model receives the parent context along with its child segments and is prompted to assign non-negative scores across the children. These local scores are normalized over the sibling set, and the final allocation for each parent is computed by averaging across accepted samples. In all experiments, we draw five samples per query and permit up to five retries per sample to recover complete, parseable outputs in cases where a response is incomplete or malformed. We evaluate the framework across eight local LLMs: \texttt{gemma2:2b} \cite{gemma2}, \texttt{gemma3:4b} \cite{team2025kimi}, \texttt{llama3.2:3b}, \texttt{llama3.2:1b} \cite{llama3}, \texttt{qwen2.5:3b} \cite{yang2025qwen3}, \texttt{qwen3:1.7b}, \texttt{qwen3:4b} \cite{qwen3technicalreport}, and \texttt{phi3:medium} \cite{abdin2024phi}.

\subsection{Evaluation Setup}
\label{app:eval-protocol}

We evaluate our framework on a set of 50 annotated papers using three independent annotation sources: \emph{human experts}, Claude Sonnet 4.6, and GPT-OSS-120B, with the latter two accessed via API calls.

Human annotators are provided with detailed annotation guidelines, described in Appendix~\ref{app:annot}. Because constructing a complete annotation of the paper hierarchy, from high level sections down to individual citations, is a time consuming task, human annotators are asked to provide annotations at a more targeted level. Specifically, they assign contribution scores to the paper's high level sections and identify the four cited papers they consider most foundational to the work. In contrast, API based language models do not face the same annotation time constraints, allowing us to obtain scores throughout the full hierarchical structure of each paper, including citation level annotations.

\paragraph{Section-level evaluation.}
For each paper we compare the model's predicted section-weight vector
$\hat{w}$ against the reference vector $w^*$ using three complementary
metrics. First, to assess whether the oracle satisfies the assumption
introduced in Section~\ref{sec:weights}, we report the $L_1$
distance $\|\hat{w} - w^*\|_1$, which captures the total absolute
discrepancy in how importance mass is distributed across
sections~\cite{gibbs2002}. Second, since the oracle is designed to produce
comparative rather than absolute scores (Equation~\ref{eq:local-normalization}),
we assess ranking quality via Spearman $\rho$~\cite{spearman1904} and
Kendall $\tau_b$~\cite{zhou-etal-2024-llm}, both of which measure
agreement in the relative ordering of section importance rather than in
exact score magnitudes~\cite{paper-eval}.
Higher values are better for $\rho$ and $\tau_b$, and lower values are better for $L_1$.

\paragraph{Citation-level evaluation.}
For the 50-paper dataset, we evaluate citation-level agreement from both distributional and ranking perspectives. We first compare the complete distribution of citation contribution scores produced by our framework, $\sigma_d(c)$, with the corresponding citation contribution scores generated by Claude Sonnet 4.6. After normalizing the scores over the cited papers in each document, we measure distributional similarity using KL divergence and Jensen-Shannon divergence (JSD)~\cite{kld, jsd}. KL divergence provides a sensitive measure of mismatch between the predicted and reference distributions, particularly when the reference assigns substantial importance to a citation that receives little weight under the predicted distribution. We additionally report JSD because it is symmetric and bounded, which provides a more stable measure of the overall difference between the two citation-score distributions. Together, these metrics evaluate whether the model recovers not only the most important citations but also the broader allocation of contribution mass across all cited papers.

We complement this distributional analysis with ranking and retrieval-based evaluation. Spearman's $\rho$ and Kendall's $\tau$ measure agreement in the relative ordering of citations, while $L_1$ distance captures differences in the assigned score mass. To evaluate agreement among the most influential citations, we compare the top-$k$ predicted citations, for $1 \leq k \leq 8$, with the four citations identified by the Claude Sonnet 4.6 annotator as most foundational to the paper. We report Recall@$k$, Hit@$k$, which indicates whether at least one reference citation is retrieved, Overlap@$k$, and nDCG@$k$. These metrics provide complementary views of citation-level agreement. Rank-based metrics assess relative ordering, top-$k$ metrics assess recovery of the most important citations, and KL divergence and JSD assess similarity across the full citation contribution distribution.

\paragraph{Baselines.}
Prior work on citation importance has largely relied on section-based
heuristics, classifying citations according to the section in which they
appear without producing continuous importance scores or quantitative
rankings~\cite{cca}. Since no directly comparable numerical baselines
exist in the literature, we design two reference baselines that require
no language model calls, thereby isolating the contribution of LLM-based
section scoring from surface-level structural signals.
The \emph{citation\_frequency} baseline assigns equal weight to every
section at each level of the hierarchy ($\frac{1}{n}$ for each of $n$ top-level
sections), distributing citation scores within each leaf section
proportionally to raw mention counts.
The \emph{length\_weighted\_frequency} baseline replaces the uniform section weights with token counts, with paragraph weights proportional to $\text{tokens} + \alpha \cdot \text{mentions}$ and citation weight $\frac{\text{mentions}}{\text{tokens}}$.\footnote{$\alpha$ is a scaling factor.}

%% file: appendix/extended-exp.tex
\section{Extended Experiment Results}\label{app:exp:extended}

\begin{figure*}[t]
\centering
\includegraphics[width=1\textwidth]{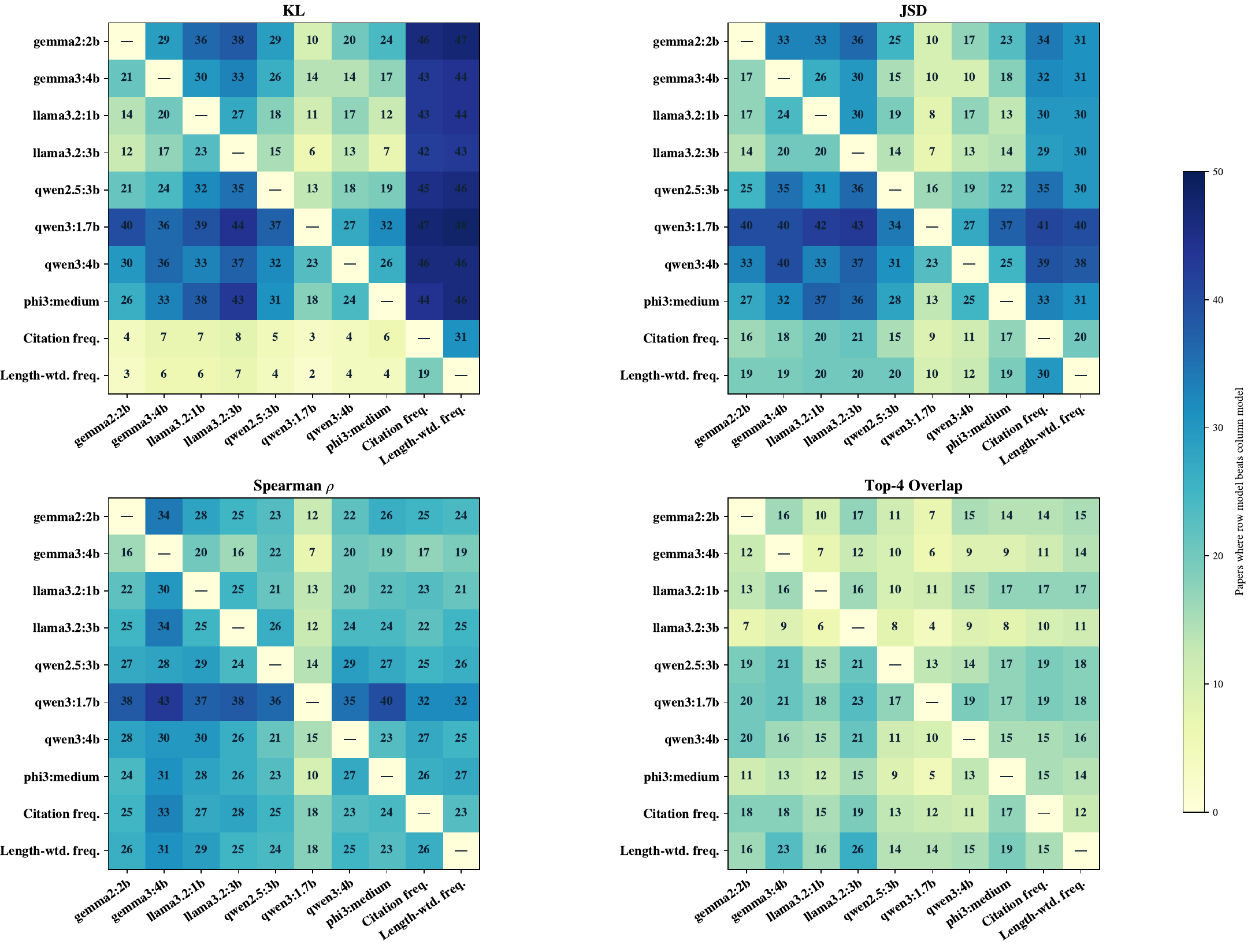}
\vspace{-3mm}
\caption{Pairwise citation-level comparisons across 50 papers. \textbf{Each cell reports the number of papers for which the row method outperforms the column method, with ties not counted.} A win corresponds to lower KL divergence and JSD or higher Spearman's $\rho$ and Top-4 overlap. Darker cells indicate a larger number of pairwise wins. \texttt{qwen3:1.7b} achieves the highest total number of wins across all four metrics, with its strongest advantage appearing on the distributional metrics KL divergence and JSD.}
\label{fig:pairwise_winrates}
\vspace{-4mm}
\end{figure*}

\begin{figure*}[t]
\centering
\includegraphics[width=1\textwidth]{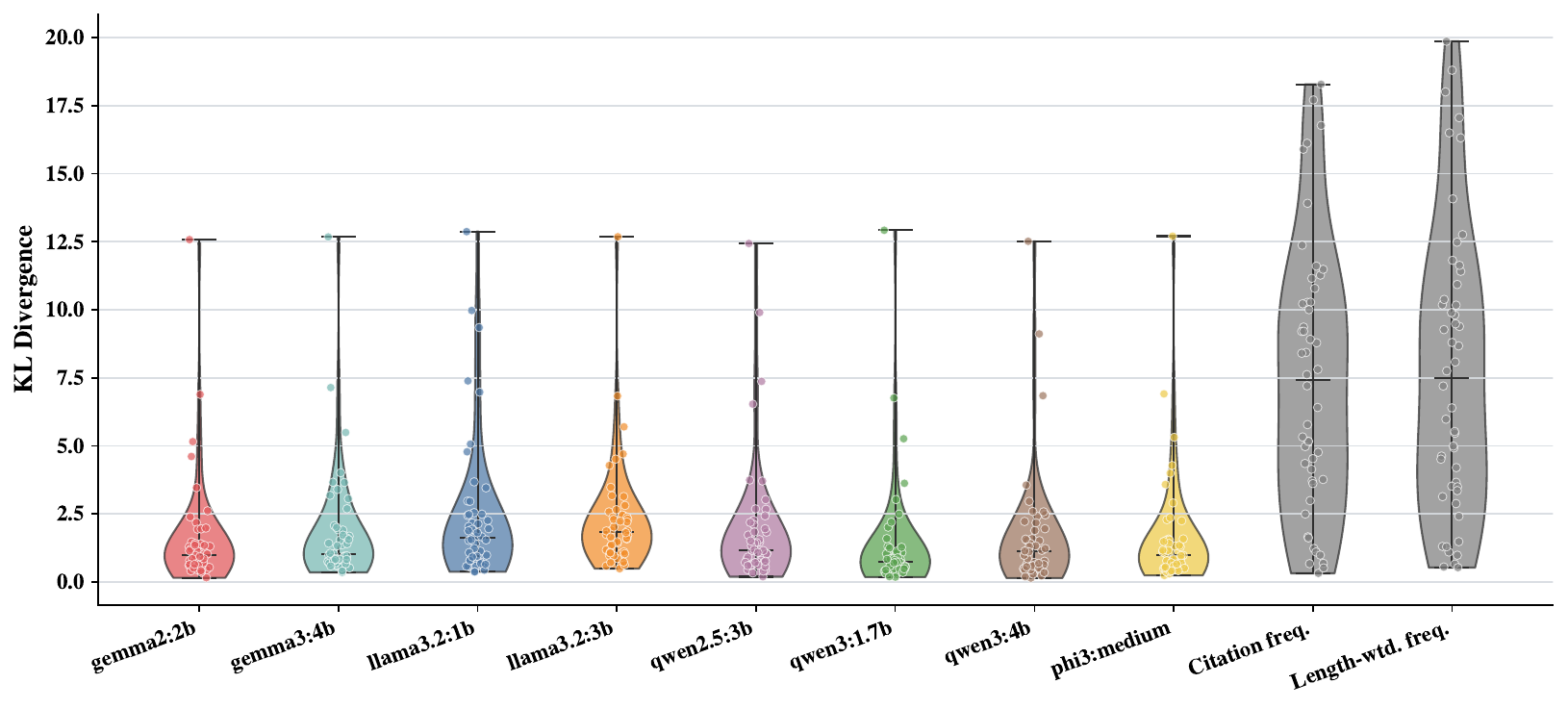}
\vspace{-3mm}
\caption{Distribution of per-paper KL divergence between citation contribution score distributions produced by each model or baseline and the Claude Sonnet 4.6 reference across 50 papers. Lower values indicate closer distributional agreement, while the vertical spread and upper tails reflect variation in agreement across papers. \texttt{qwen3:1.7b} achieves the lowest mean KL divergence among the evaluated models, while the \emph{citation-frequency} and \emph{length-weighted frequency} baselines exhibit substantially higher divergence and heavier upper tails.}
\label{fig:citation-kl-violin}
\vspace{-4mm}
\end{figure*}

\subsection{Section ranking}\label{app:sec-ranking}

Figures~\ref{fig:section_validation} in Section~\ref{subsec:validation} report validation results across the 50 annotated papers.

Among LLMs, \texttt{qwen3:1.7b} is the strongest section-ranking model against human annotations, achieving the highest Spearman correlation ($\rho = 0.669$), the highest Kendall tau-b ($\tau_b = 0.551$), and the lowest $L_1$ error (0.505). It is followed by \texttt{qwen2.5:3b} ($\rho = 0.616$, $\tau_b = 0.509$, $L_1 = 0.530$) and \texttt{llama3.2:3b} ($\rho = 0.594$, $\tau_b = 0.495$, $L_1 = 0.604$). This makes \texttt{qwen3:1.7b} the most human-aligned LLM in both rank-based and absolute-distribution agreement, a notable result given its relatively small size. At the other end, \texttt{llama3.2:1b} performs worst across the human-alignment metrics ($\rho = 0.270$, $\tau_b = 0.232$, $L_1 = 0.650$), while \texttt{gemma3:4b} also trails the stronger models on rank correlation ($\rho = 0.445$, $\tau_b = 0.359$), indicating that larger parameter count alone does not guarantee better section-importance calibration at this quantization level.

The same overall ordering persists when model outputs are compared against the two API references. Against Sonnet-4.6, \texttt{qwen3:1.7b} remains the best LLM on all three section metrics ($\rho = 0.612$, $\tau_b = 0.498$, $L_1 = 0.527$), with \texttt{llama3.2:3b} and \texttt{qwen2.5:3b} forming the next tier. Against GPT-OSS-120B, \texttt{qwen3:1.7b} again leads all LLMs ($\rho = 0.555$, $\tau_b = 0.467$, $L_1 = 0.438$), followed by \texttt{qwen2.5:3b} and \texttt{llama3.2:3b}. This consistency across human, Sonnet-4.6, and GPT-OSS-120B suggests that the relative ranking differences are robust rather than artifacts of a single reference source.

The \emph{length-weighted frequency} baseline remains a strong non-LLM heuristic. It is the best overall system by Spearman and Kendall when evaluated against Sonnet-4.6 ($\rho = 0.760$, $\tau_b = 0.644$) and GPT-OSS-120B ($\rho = 0.659$, $\tau_b = 0.557$), and it also attains the lowest $L_1$ against Sonnet-4.6 (0.439). Against human annotations, however, \texttt{qwen3:1.7b} remains stronger on rank correlation, and also yields the lowest LLM $L_1$. This pattern suggests that section length is indeed a powerful structural prior, but that it only partially explains human judgments: it captures some broad tendencies shared by API references, yet it does not fully replace content-sensitive modeling when the goal is alignment with human section-importance annotations.

\begin{figure*}[t]
\centering
\includegraphics[width=0.8\textwidth]{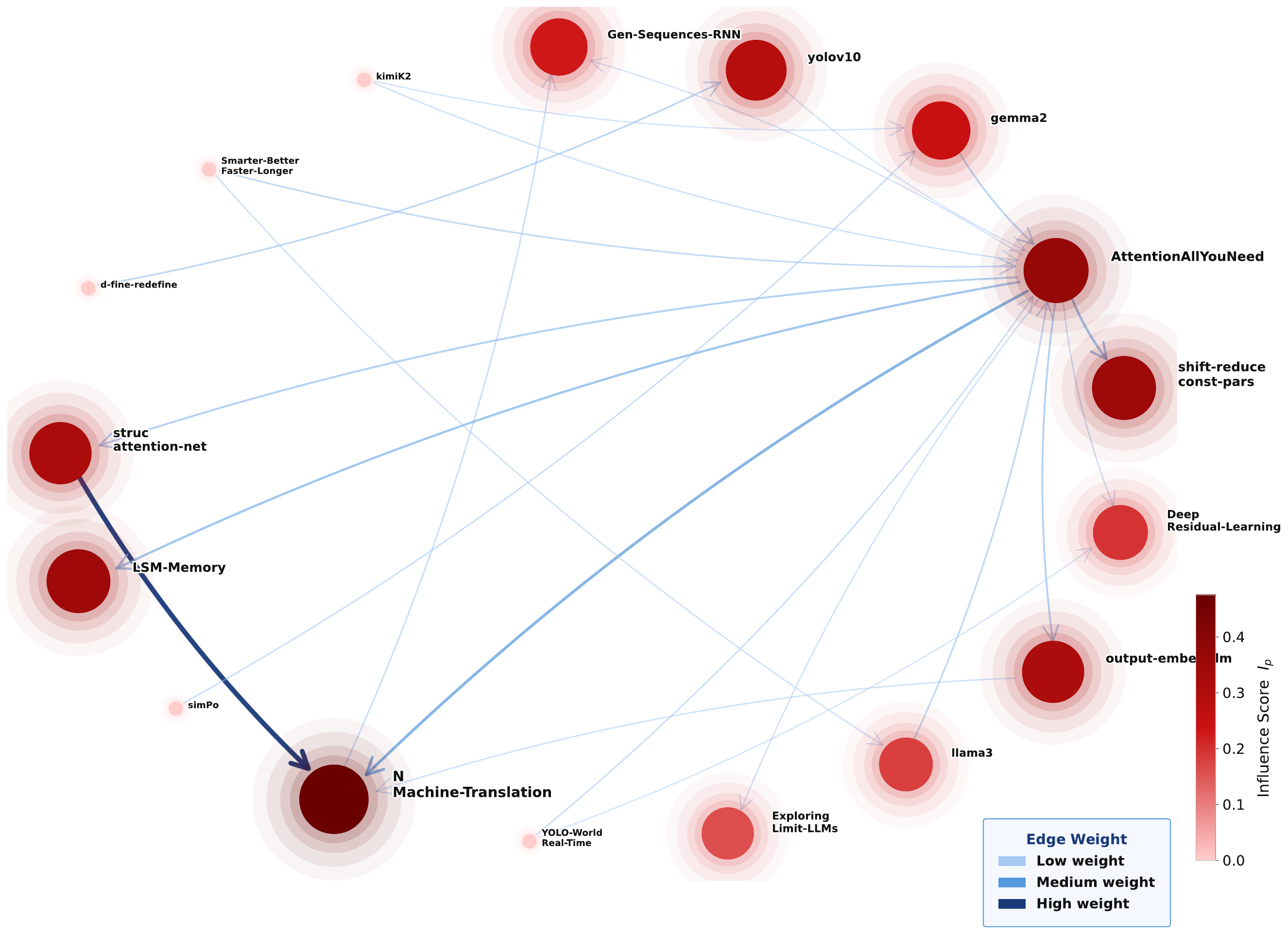}
\caption{Citation network of the case-study corpus. Each node represents a paper in the corpus, with node size proportional to the normalized cross-paper influence score $I_{\mathcal{P}}(p)$, which captures how influential a paper is within the network through direct and indirect citations. A directed edge from paper $a$ to paper $b$ indicates that $a$ cites $b$.}
\label{fig:citation_network_Ip}
\end{figure*}

\subsection{Ablation Study: Effect of Model Size}\label{app:exp:modelsize}
To further examine whether contribution-estimation quality improves with model scale, we compare the models using paper-level pairwise wins. For each pair of methods, Fig.~\ref{fig:pairwise_winrates} reports the number of papers for which one method outperforms the other over our 50-paper dataset. A win
corresponds to lower KL divergence or JSD and higher Spearman's $\rho$ or Top-4 overlap. Ties are excluded.

\texttt{qwen3:1.7b} shows the strongest overall performance. Compared with the other evaluated LLMs, it accumulates 255 wins for KL divergence, 263 for JSD, 267 for Spearman's $\rho$, and 135 for Top-4 overlap, the highest total for every metric. Its advantage extends beyond mean performance and is observed consistently across individual papers. This result is particularly notable because \texttt{qwen3:1.7b} is one of the smallest models in the comparison. Therefore, it provides a favorable combination of model size and contribution scoring quality.

The pairwise results also demonstrate that increasing parameter count does not consistently improve performance. For Qwen3, \texttt{qwen3:1.7b} outperforms \texttt{qwen3:4b} on 27 versus 23 papers for both KL divergence and JSD, 35 versus 15 for Spearman's $\rho$, and 19 versus 10 for Top-4 overlap.
A similar pattern appears for Llama 3.2. The 1B model outperforms the 3B model by 27 to 23 on KL divergence, 30 to 20 on JSD, and 16 to 6 on Top-4 overlap, while the two models are tied at 25 wins each for Spearman's $\rho$.

The Llama result is also interesting when considered together with the
section-level evaluation (Figure~\ref{fig:section_validation}). There, the 3B model performs substantially better than the 1B model, whereas the citation scoring results favor the smaller
model on three of the four pairwise metrics. This suggests that the effect of model scale may depend on the granularity of the contribution estimation task. A model that is better at allocating contribution across broad sections is not necessarily better at distinguishing contribution among individual citations.

The Gemma comparison follows the same general pattern. \texttt{gemma2:2b} outperforms \texttt{gemma3:4b} by 29 to 21 on KL divergence, 33 to 17 on JSD, 34 to 16 on Spearman's $\rho$, and 16 to 12 on Top-4 overlap. However, because the two Gemma models also differ in model generation and training, this comparison does not isolate parameter count alone. It should therefore be interpreted as additional evidence that larger models do not automatically provide better contribution estimates rather than as a controlled scaling experiment.

Overall, parameter count is not a reliable predictor of contribution scoring quality in our experiments. Performance depends on both the model and the level at which contribution is estimated, and the strongest model is not the largest one. The consistent advantage of \texttt{qwen3:1.7b} across
distributional, ranking, and top-citation agreement supports its selection as the backbone for the corpus-level analysis. Because ties are excluded from the pairwise counts, absolute totals should not be compared directly across metrics, particularly for Top-4 overlap, which produces more ties than the
continuous metrics.




\subsection{Additional Corpus-level Case Study Details}
\label{app:corpus-case}

Figure~\ref{fig:citation_network_Ip} visualizes the weighted citation graph used in the corpus-level case study. Each directed edge follows the citation direction from a citing paper to a cited paper, and edge weights are derived from the contribution assigned to the cited work by the citing paper. Node size is proportional to the normalized influence score $I_{\mathcal{P}}(p)$, which measures the share of cross-paper influence mass attributed to each paper after removing self-contribution.

The figure highlights the difference between raw citation degree and propagated influence. Papers with the same in-degree number within the corpus may receive different $I_{\mathcal{P}}$ values because influence depends not only on how many papers cite them, but also on the contribution weights of those citations and the influence of the citing papers. For example, papers with a single incoming citation can differ substantially in propagated influence when they are cited by papers occupying different positions in the weighted citation DAG.

The case study also illustrates the distinction between corpus-level contribution $\sigma_{\mathcal{P}}(p)$ and normalized influence $I_{\mathcal{P}}(p)$. The former combines a paper's intrinsic original contribution with propagated credit received through citation paths, whereas the latter isolates cross-paper influence by removing the paper's self-contribution. As a result, recent or terminal papers can obtain high $\sigma_{\mathcal{P}}(p)$ while having low or zero $I_{\mathcal{P}}(p)$, whereas foundational papers can accumulate high influence through direct and indirect citation paths even when they do not have the highest corpus-level contribution score.

Finally, the case study should be interpreted as an induced-corpus analysis rather than a global impact ranking. The scores depend on the selected subgraph: omitting relevant descendants or neighboring papers can reduce a paper's propagated influence, and a small corpus may compress or distort citation ratios observed in the global scholarly graph. Thus, $\sigma_{\mathcal{P}}(p)$ and $I_{\mathcal{P}}(p)$ are best understood as local structural measurements conditioned on the analyzed corpus.

%% file: appendix/human-noise.tex
\subsection{Human Annotation Consistency and Variability}\label{app:human-noise}

Contribution judgments are inherently subjective, and repeated assessments of the same paper may not produce identical scores or citation selections. We therefore conduct two complementary analyses to characterize variability in the human annotations. First, we collect a second round of high-level section
scores and directly measure agreement with the original annotations. Second, we examine whether annotators reconsider their original top-4 foundational citations after reviewing model-generated citation suggestions. The original annotations are retained unchanged and are used throughout all reported
evaluations.

\paragraph{Section-level consistency.}
For each of the 50 papers, annotators completed a second round of high-level
section scoring. We compare the two annotation rounds using Spearman's
$\rho$, Kendall's $\tau_b$, $L_1$ distance, and Jensen-Shannon divergence
(JSD). Table~\ref{tab:human_round_consistency} summarizes the resulting
test-retest agreement.

\begin{table}[t]
\centering
\scriptsize
\setlength{\tabcolsep}{5pt}
\begin{tabular}{lcc}
\toprule
Metric & Mean & 95\% CI \\
\midrule
Spearman $\rho$ $\uparrow$
    & 0.891 & $[0.856,\,0.923]$ \\
Kendall $\tau_b$ $\uparrow$
    & 0.822 & $[0.778,\,0.864]$ \\
$L_1$ $\downarrow$
    & 0.253 & $[0.220,\,0.288]$ \\
JSD $\downarrow$
    & 0.015 & $[0.012,\,0.018]$ \\
\bottomrule
\end{tabular}
\caption{Test-retest agreement between two rounds of human high-level
section annotations over 50 papers. Confidence intervals are 95\% bootstrap
intervals over papers. Higher values indicate stronger agreement for
Spearman's $\rho$ and Kendall's $\tau_b$, while lower values indicate
stronger agreement for $L_1$ and JSD.}
\label{tab:human_round_consistency}
\end{table}

\begin{table*}[t]
\centering
\small
\setlength{\tabcolsep}{10pt}
\begin{tabular}{lcc}
\toprule
Metric & Human R1 vs R2 & Sonnet-4.6 R1 vs R2 \\
\midrule
Mean Spearman $\rho$ $\uparrow$ & 0.891 {\scriptsize[$0.856,\,0.922$]} & 0.952 {\scriptsize[$0.927,\,0.970$]} \\
Mean Kendall $\tau_b$ $\uparrow$ & 0.822 {\scriptsize[$0.777,\,0.863$]} & 0.914 {\scriptsize[$0.886,\,0.937$]} \\
Mean $L_1$ $\downarrow$ & 0.253 {\scriptsize[$0.220,\,0.287$]} & 0.129 {\scriptsize[$0.098,\,0.168$]} \\
Mean JSD $\downarrow$ & 0.015 {\scriptsize[$0.012,\,0.018$]} & 0.009 {\scriptsize[$0.003,\,0.019$]} \\
Top-1 agreement $\uparrow$ & 0.740 {\scriptsize[$0.620,\,0.860$]} & 0.920 {\scriptsize[$0.840,\,0.980$]} \\
Top-2 overlap $\uparrow$ & 0.830 {\scriptsize[$0.760,\,0.890$]} & 0.910 {\scriptsize[$0.850,\,0.960$]} \\
\bottomrule
\end{tabular}
\caption{Round to round agreement at the section level, comparing two rounds of human annotation against two rounds of Claude Sonnet-4.6 annotation over the same 50 papers. Entries report means with 95\% bootstrap confidence intervals. Higher is better for Spearman $\rho$, Kendall $\tau_b$, top-1 agreement, and top-2 overlap; lower is better for $L_1$ and JSD.}
\label{tab:human-vs-sonnet-annotation-rounds}
\end{table*}
Human section rankings are highly stable across annotation rounds. The two rounds achieve a mean Spearman correlation of $\rho=0.891$ and Kendall correlation of $\tau_b=0.822$, with median values of $0.915$ and $0.833$, respectively. In 42 of the 50 papers, the round to round Spearman correlation
is at least $0.8$, and in 29 papers it is at least $0.9$. The highest-scored section also remains unchanged for 74\% of the papers, while the mean overlap between the two highest-scored sections is $0.83$.

Agreement remains strong but exhibits greater variation when considering the precise allocation of contribution mass. The two annotation rounds have a mean $L_1$ distance of $0.253$ and a mean JSD of $0.015$. These results suggest that the primary source of annotator variability is generally not a substantial change in which sections are considered important, but rather variation in how much contribution weight is assigned to them. Section-level contribution judgments therefore appear more reproducible as relative rankings than as precisely calibrated score allocations.

This result also provides useful context for the section-level validation in Figure~\ref{fig:section_validation}. Human test-retest agreement is higher than the agreement between the original human annotations and either API-based reference. The two human rounds achieve $\rho=0.891$, $\tau_b=0.822$, and $L_1=0.253$, compared with $\rho=0.865$, $\tau_b=0.779$, and $L_1=0.320$ between the original human annotations and Claude Sonnet 4.6, and $\rho=0.799$, $\tau_b=0.701$, and $L_1=0.348$ for GPT-OSS-120B (Table~\ref{tab:api_human_alignment}). Thus, the human annotations are
substantially reproducible but not deterministic, and the observed round to round variation provides an empirical indication of the uncertainty associated with assigning exact contribution weights.

Interestingly, the API annotator is even more reproducible across repeated
annotations than the human annotators. As shown in Table~\ref{tab:human-vs-sonnet-annotation-rounds}, Claude Sonnet 4.6 achieves higher round to round rank agreement, with Spearman $\rho=0.952$ and Kendall
$\tau_b=0.914$, compared with $0.891$ and $0.822$ for the human annotations.
It also exhibits substantially lower variation in the assigned score mass, with an $L_1$ distance of $0.129$ compared with $0.253$ for humans.
Top-1 agreement similarly increases from $0.740$ to $0.920$, while Top-2 overlap increases from $0.830$ to $0.910$. These results indicate that, under a fixed annotation protocol, the LLM produces more reproducible section-level judgments and is less sensitive to repeated evaluation than human annotators.

This greater consistency should not be interpreted as evidence that the LLM annotations are necessarily more accurate. Reproducibility measures the stability of repeated judgments, whereas validity concerns whether those judgments reflect the underlying scientific contribution. A model may produce highly consistent annotations while also reproducing the same systematic bias across runs. The result therefore supports the use of LLMs as stable contribution estimators, while human annotations remain important as an independent reference for evaluating alignment with expert judgment.

\paragraph{Citation-level variability.} We additionally examine the stability of the human top-4 citation selections through a post-hoc review. After the original annotations had been collected, annotators were shown four highly ranked model-generated citations and asked whether, given an opportunity to reconsider their selections, they would revise their original top-4 set.\footnote{The original annotations were not modified and remain the reference annotations used in all reported experiments. This preserves the independence of the evaluation labels from model-generated outputs.}

In approximately 78\% of cases, annotators indicated that they would include at least one model-suggested reference in a revised selection. In approximately 21\% of cases, annotators explicitly reported having overlooked one or more important references during their initial assessment, and in approximately 32\% of cases the model's highest-ranked citation was subsequently judged appropriate for inclusion. These observations suggest that selecting a small, fixed set of foundational citations is less stable than assigning a broad ordering of section importance. A single-pass top-4 annotation may omit references that the same expert considers important upon closer reconsideration.

This post-hoc analysis should not be interpreted as an independent estimate of annotation error. Because annotators were exposed to model-generated suggestions before reconsidering their selections, some revisions may reflect anchoring bias or suggestion effects rather than genuine corrections of omissions. For this reason, all quantitative evaluations use only the original annotations, which were collected before annotators were exposed to the model-generated recommendations.

Independent API-based annotations provide complementary evidence that citation-level selection admits greater disagreement than section-level ranking. As shown in Table~\ref{tab:api_human_alignment}, Claude Sonnet 4.6 and GPT-OSS-120B exhibit strong agreement with the human high-level section annotations, with Spearman correlations of $0.865$ and $0.799$, respectively. Agreement with the human top-4 citation selections is considerably lower. Claude Sonnet 4.6 recovers on average $1.72$ of the four human selected citations, corresponding to Recall@4 of $0.430$, while GPT-OSS-120B recovers $1.32$, corresponding to Recall@4 of $0.330$. Nevertheless, Hit@4 remains high at $0.900$ and $0.840$, respectively, indicating that the API-generated top-4 sets usually contain at least one citation also selected by the human annotator.

The large gap between Hit@4 and Recall@4 suggests that disagreement often concerns which subset of several plausible foundational references should occupy a small top-4 set, rather than a complete failure to identify relevant prior work. This interpretation is also consistent with the post-hoc human revisions, where annotators frequently considered model suggested references
reasonable additions to their original selections. Consequently, Recall@4 and nDCG@4 should be interpreted as measures of agreement with one independently elicited expert selection rather than exact measures of scientific correctness. They may penalize model selected citations that are scientifically defensible but absent from the annotator's original top-4 set.

\paragraph{Implications for evaluation.}
The two analyses reveal different forms of annotation variability. High-level section judgments are highly reproducible in their relative ordering, although the precise allocation of contribution mass varies across repeated assessments. Citation-level judgments are less constrained because annotators must select a small subset from a potentially larger set of plausibly foundational references. We therefore treat human annotations as an independent evaluation reference rather than an exact ground truth. The use of multiple reference sources, complementary ranking and distributional metrics, and the preservation of the original annotations reduces the dependence of our conclusions on any single annotation set.

%% file: appendix/prompts.tex
\section{LLM Prompts}
\label{app:prompts}

The following prompts are used at each stage of the scoring pipeline.
Each stage shows the \textbf{system prompt} (blue) followed by the
\textbf{user prompt template} (teal).
Placeholders in curly braces (e.g.\ \texttt{\{parent\_name\}}) are
filled at runtime with values extracted from the paper.

\subsection{Stage 1: Section Scoring}

\begin{tcolorbox}[systemprompt, title=System Prompt]
You are an expert academic reviewer.
Given the child sections or subsections of one paper segment, distribute
100 points across them based on how much each contributes to the parent
segment's main scientific contribution. Treat section names as meaningful context. Use the excerpts and parent
context to judge which children contribute most strongly through methods,
theory, algorithms, proofs, experimental findings, analysis, problem
formulation, or other core scientific content.
Background, transitions, setup details, summaries, and lower-impact
narrative should usually receive less weight.
Return percentages, not copied numbers from the excerpts.
Percentages must be non-negative and sum to 100.
Return plain text lines only. Do not output JSON.
\end{tcolorbox}

\begin{table*}[!tbh]
\centering
\small
\begin{tabular}{clll}
\hline
\textbf{Stage} & \textbf{Name} & \textbf{Input} & \textbf{Output} \\
\hline
1 & Section scoring     & Section tree + excerpts  & \% weight per section \\
2 & Subsection scoring  & Parent + children        & \% weight per child   \\
3 & Paragraph split     & Paragraph text           & technical / citation \% \\
4 & Score split         & \% from Stage 3 (code)   & $\sigma_\text{orig}$, $\sigma_\text{cit}$ per paragraph \\
5 & Citation allocation & Paragraph + citation list & \% weight per citation \\
\hline
\end{tabular}
\caption{LLM prompt stages in the scoring pipeline.}
\label{tab:prompt-stages}
\end{table*}

\begin{tcolorbox}[userprompt, title=User Prompt Template]
Parent segment: \{parent\_name\}\\
Parent score to distribute: \{parent\_score\}\\
Parent context: these items are section/subsection children of
"\{parent\_name\}". \\
Treat their names as meaningful context, not just labels.\\[4pt]
Parent content (context): \{parent\_content\}\\[4pt]
Task: Assign percentage importance to the following child segments
according to how much each contributes to the parent segment's main
scientific contribution.\\[4pt]
Guidelines:\\
- Higher scores: segments that carry the core scientific contribution
  through methods, theory, algorithms, proofs, experimental findings,
  analysis, or important problem formulation.\\
- Lower scores: segments that mainly provide background, motivation,
  transitions, setup details, or lower-impact narrative.\\[4pt]
Child segments (id $\to$ name and excerpt): \{items\}\\[4pt]
Output one percentage per child. Lines must sum to 100.
Do not output JSON. Do not include explanations.
\end{tcolorbox}

\subsection{Stage 2: Subsection and Child Scoring}
Uses the same system and user prompt templates as Stage 1,
applied recursively to each parent-children group in the section tree.

\subsection{Stage 3: Paragraph Scoring (Technical vs.\ Citation Split)}

In the implementation and prompting framework, the term ``technical\_percentage'' is used interchangeably with ``original contribution'' as defined in the paper terminology.
\begin{tcolorbox}[systemprompt, title=System Prompt]
You are an expert academic reviewer. For each paragraph, split its total
score into two components:

\textbf{technical\_percentage}: the percentage of value this paragraph
contributes through the authors' own original ideas — new definitions,
algorithms, proofs, experimental design, results, or analysis that would
survive even if all cited works were removed.

\textbf{citation\_percentage}: the percentage of value this paragraph
derives from cited prior work — including summarizing, comparing,
contextualizing, or building upon existing work.

technical + citation must equal 100 for each paragraph.
\end{tcolorbox}

\begin{tcolorbox}[userprompt, title=User Prompt Template]
Task: For each paragraph, split the value into technical\_percentage
and citation\_percentage.\\[4pt]
Section context: these paragraphs are from the "\{section\_name\}"
section. Use this to calibrate — a Related Work paragraph behaves
very differently from a Methods paragraph.\\[4pt]
Calibration guidance:\\
- Related Work (summarizes prior methods): citation $\approx$ 80-95\%\\
- Related Work (novel connection to the paper's gap): citation $\approx$ 50-70\%\\
- Methods (authors' new algorithm): citation $\approx$ 0--15\%\\
- Conclusion (paper's own contributions): citation $\approx$ 0--10\%\\[4pt]
Rules: technical + citation = 100; if has\_citations is false,
citation\_percentage = 0.\\[4pt]
Paragraph entries: \{paragraphs\_json\}\\[4pt]
Output: \texttt{paragraph\_id: technical=<float>, citation=<float>}
\end{tcolorbox}

\subsection{Stage 4: Splitting Score: Citation vs.\ Originality}
\textit{The citation\_percentage and technical\_percentage output from
Stage 3 are used directly to divide each paragraph's importance score
into a citation channel score and a technical (originality) channel
score in code, without an additional LLM call.}

\subsection{Stage 5: Distributing Citation Score Among Citations}

\begin{tcolorbox}[systemprompt, title=System Prompt]
You are an expert academic reviewer.
Your task is to distribute percentage importance among the citations
that appear in a paragraph, based on how much each cited work
contributes to the paragraph's claims or grounding.\\
Percentages must be non-negative and sum to 100.\\
Return plain text lines only.
\end{tcolorbox}

\begin{tcolorbox}[userprompt, title=User Prompt Template]
Paragraph id: \{paragraph\_id\}\\
Paragraph citation score to distribute: \{paragraph\_citation\_score\}\\[4pt]
Paragraph text: \{paragraph\_text\}\\[4pt]
Task: Divide percentage importance among the citations below.\\
- Higher share: citation directly supports the paragraph's claim;
  provides key comparison or baseline; introduces a method the
  paragraph builds upon.\\
- Lower share: citation is peripheral or only mentioned briefly.\\[4pt]
Citation entries (citation\_id $\to$ citation and context):
\{citations\_json\}\\[4pt]
Output one line per citation.
Lines must sum to 100.
Do not output JSON. Do not include explanations.
\end{tcolorbox}

%% file: appendix/annotators.tex
\clearpage
\onecolumn
\subsection{Human Annotation Guidelines}
\label{app:annot}

Annotators described in Section~\ref{subsec:validation} were either authors of the annotated papers or scholars with sufficient familiarity with the relevant research area to assess them, rather than crowd annotators. Each annotator completed two tasks: (1) distributing a unit contribution score across the paper's high-level sections, and (2) identifying the four cited papers most foundational to the work.
The following guidelines were provided to the human annotators.

\begin{tcolorbox}[
  guidelinebox,
  title={Guidelines}
]

\textbf{What to score:}

Score only sections and subsections.

Do NOT score paragraphs.

Do NOT separately score citations inside sections.

\medskip

\textbf{General principle}

Treat the paper as a hierarchy.

At each parent node, distribute a fixed budget across its immediate child sections.

You should judge how much ORIGINAL SCIENTIFIC VALUE each child section creates for the paper.

\medskip

Ask yourself: ``If this section were removed, how much of the paper's scientific contribution would be lost?''

\medskip

\textbf{Focus on:}
\begin{itemize}
    \item new methods
    \item new algorithms
    \item new theory or proofs
    \item new formulations
    \item new experimental design
    \item new findings or results
    \item new technical analysis
\end{itemize}

\textbf{Give lower scores to sections that mainly:}
\begin{itemize}
    \item motivate the problem
    \item provide background
    \item summarize related work
    \item restate results from other sections
    \item provide transitions or narrative framing
\end{itemize}

\medskip
\begin{center}
    \rule{0.9\textwidth}{0.4pt}
\end{center}

\textbf{How to score}

\medskip

\textbf{Step 1: Top-level sections}

Start with the whole paper budget = 100 points.

Distribute those 100 points across the top-level sections of the paper.

Example:

\begin{quote}
Introduction: 12

Methodology: 36

Experiments: 34

Related Work: 8

Conclusion: 5

Limitations: 5
\end{quote}

The top-level section scores must sum to exactly 100.

\medskip
\begin{center}
    \rule{0.9\textwidth}{0.4pt}
\end{center}

\textbf{What counts as high importance}

\medskip

\textbf{Higher scores:}
\begin{itemize}
    \item the core technical method of the paper
    \item the main model, algorithm, or framework
    \item the main theorem/proof/formulation
    \item the experiments that establish the paper's core claims
    \item the analysis that is essential to the contribution
\end{itemize}

\textbf{Lower scores:}
\begin{itemize}
    \item background or setup
    \item narrative overview
    \item broad motivation
    \item literature review / related work summary
    \item discussion that mainly repeats earlier claims
    \item conclusion text that summarizes rather than adds new scientific value
\end{itemize}

\end{tcolorbox}

\begin{tcolorbox}[
  foundationalbox,
  title={PART 2: TOP 4 FOUNDATIONAL PAPERS}
]

\textbf{Goal}

Choose the 4 cited papers that are most foundational to your work.

\medskip

\textbf{Definition of ``foundational'':}

A foundational paper is one that directly supports the core technical contribution of your paper.

\textbf{Examples of foundational use:}
\begin{itemize}
    \item a method your work builds on
    \item a baseline that is central to your evaluation
    \item a theory or formulation your work depends on
    \item a prior framework that strongly shaped your design
    \item a key paper without which your contribution would be substantially different
\end{itemize}

\textbf{Examples of non-foundational use:}
\begin{itemize}
    \item broad background references
    \item generic motivation references
    \item citations included mainly for completeness
    \item peripheral comparisons
    \item papers mentioned only in passing
\end{itemize}

\end{tcolorbox}